\documentclass[11pt]{article}
\usepackage{graphicx}
\usepackage{xcolor}
\usepackage{amsmath}
\usepackage{amssymb}
\usepackage{url}
\usepackage{algorithm}
\usepackage{algpseudocode}
\usepackage{booktabs}
\usepackage{xcolor}
\usepackage{fullpage}
\usepackage{amsthm}
\usepackage{mathtools}
\usepackage{subcaption}
\usepackage{float}
\usepackage{hyperref}
\usepackage{cite}

\newtheorem{theorem}{Theorem}       
\newtheorem{lemma}{Lemma}

\theoremstyle{remark}
\newtheorem{remark}{Remark}
\theoremstyle{plain}
\newtheorem{proposition}{Proposition}

\usepackage[nameinlink,capitalise]{cleveref}
\crefname{figure}{Fig.}{Figs.}
\Crefname{figure}{Fig.}{Figs.}
\crefname{equation}{Eq.}{Eqs.}
\Crefname{equation}{Eq.}{Eqs.}
\crefname{algorithm}{Alg.}{Algs.}
\Crefname{algorithm}{Alg.}{Algs.}
\crefname{remark}{Remark}{Remarks}
\Crefname{remark}{Remark}{Remarks}
\crefname{proposition}{Proposition}{Propositions}
\Crefname{proposition}{Proposition}{Propositions}

\usepackage{placeins} 

\newcommand{\fedge}{\xrightarrow{\;f\;}}
\newcommand{\bedge}{\xrightarrow{\;b\;}}

\title{Extended Serial Safety Net: A Refined Serializability Criterion for Multiversion Concurrency Control}

\usepackage{authblk}

\author[1]{Atsushi Kitazawa}
\author[2]{Chihaya Ito}
\author[3]{Yuta Yoshida\thanks{Corresponding author: \texttt{yyoshida\_20@nec.com}}}
\author[2]{Takamitsu Shioi}

\affil[1]{Independent IT Consultant, Tokyo, Japan}
\affil[2]{NEC Solution Innovators, Ltd., Tokyo, Japan}
\affil[3]{NEC Corporation, Tokyo, Japan}

\begin{document}

\maketitle

\begin{abstract}
A long line of concurrency-control (CC) protocols argues correctness via a \emph{single serialization point}
(begin or commit), an assumption that is incompatible with \emph{snapshot isolation} (SI), where read--write
anti-dependencies arise. \emph{Serial Safety Net} (SSN) offers a lightweight commit-time test but is conservative
and effectively anchored on commit time as the sole point.

We present \textbf{ESSN}, a principled generalization of SSN that \emph{relaxes} the exclusion condition to allow
more transactions to commit safely, and we prove that this preserves multiversion serializability (MVSR)
and that it strictly subsumes SSN. ESSN states an MVSG (Multiversion Serialization Graph)-based criterion and introduces 
a \emph{known total order} over transactions (KTO; e.g., begin-ordered or commit-ordered) 
for reasoning about the graph's serializability.

With a single commit-time check under invariant-based semantics, ESSN's exclusion condition preserves monotonicity 
along per-item version chains, and eliminates chain traversal. The protocol is Direct Serialization Graph (DSG)-based 
with commit-time work linear 
in the number of reads and writes, matching SSN's per-version footprint.

We also make mixed workloads explicit by defining a \emph{Long} transaction via strict interval containment of
\emph{Short} transactions, and we evaluate ESSN on reproducible workloads. Under a commit-ordered KTO, using
begin-snapshot reads reduces the long-transaction abort rate by up to $\Delta\!\approx\!0.25$ absolute (about
$50\%$ relative) compared with SSN.
\end{abstract}


\section{Introduction}
\label{sec:introduction}

A long line of concurrency-control (CC) protocols has established correctness via a
\emph{single serialization point} (e.g., begin-time or commit-time), assuming that the serialization
order of transactions is aligned with that point---especially in modern high-throughput systems
such as Silo~\cite{tu2013-silo}, Cicada~\cite{lim2017-cicada}, TicToc~\cite{yu2016-tictoc},
and CCaaLF~\cite{pan2025-ccaalf}.
In parallel, many production DBMSs adopt SI because reads observe
a consistent snapshot and avoid blocking writers, enabling high concurrency---especially for
read-heavy workloads~\cite{berenson1995-critique,cahill2009-ssi-tods}.
However, SI is not serializable: it permits read--write anti-dependencies inherent to SI that are not 
captured by the single-serialization-point reasoning prevalent in mainstream CC.

To address SI anomalies caused by those anti-dependencies,
\emph{Serializable Snapshot Isolation} (SSI)~\cite{cahill2009-ssi-tods}
and \emph{Serial Safety Net} (SSN)~\cite{wang2017-ems} were proposed for settings where the 
designated serialization point and the resulting serialization order may diverge.
SSN, in particular, introduces an explicit notion of back-edges and offers a lightweight
runtime test that detects potential MVSG cycles at commit time.
It covers not only SI but also \emph{read committed} (RC), making it a generic safety net.
Yet SSN has two drawbacks for today's diverse MVCC designs:
(i) its exclusion condition is conservative (leading to unnecessary aborts), and
(ii) its reasoning implicitly anchors on commit time, limiting applicability to begin-ordered/timestamped variants.

\paragraph{Our stance.}
We recast serializability as an MVSG-based \emph{criterion} over three declarative components:
(i) a \emph{version function} (VF) selecting the visible version for each read,
(ii) a per-item \emph{version order} (VO) that orders writes, and
(iii) a configurable \emph{known total order} (KTO) over transactions used for reasoning (e.g., begin- or commit-ordered).
This perspective is close in spirit to serializability \emph{checkers} such as COBRA~\cite{tan2020-cobra}, which encode
MVSG constraints, including the choice of VO, as an SMT instance; 
by contrast, we treat KTO as an explicit design parameter and adopt a criterion that 
\emph{aligns VF and VO with the declared KTO}.

Given $(\mathrm{VF},\mathrm{VO},\mathrm{KTO})$, we introduce \textbf{Extended Serial Safety Net (ESSN)} as a two-part framework: 
(i) a criterion requiring VF/VO to align with the declared KTO, and 
(ii) an MVSG-based exclusion test, a sufficient condition for MVSG acyclicity, that we prove sound. 
ESSN is independent of low-level implementation choices and strictly subsumes SSN. 
In the rest of this paper, ``ESSN'' refers to both the criterion and its exclusion test; unless stated otherwise, ``KTO'' denotes the chosen total order.

\paragraph{Example 1 (motivating).}
Consider schedule $M_1$ in Fig.~\ref{fig:schedule-m1}.
Under a commit-ordered KTO and a VO aligned with commits, the MVSG contains
a back-edge $t_4 \xrightarrow{b} t_2$ and a forward anti-dependency $t_3 \xrightarrow{f} t_4$.
SSN aborts $t_4$ via $\pi(t_4)\!\le\!\eta(t_4)$ despite the graph being acyclic (a false positive).
In contrast, ESSN \emph{propagates along forward-edges} to form $\xi(\cdot)$ and
accepts $M_1$; details appear in \S\ref{sec:ssn-limitations}--\ref{sec:ssn-illustrative}.

\medskip
\begin{figure}[t]
  \centering
  \includegraphics[width=0.6\linewidth]{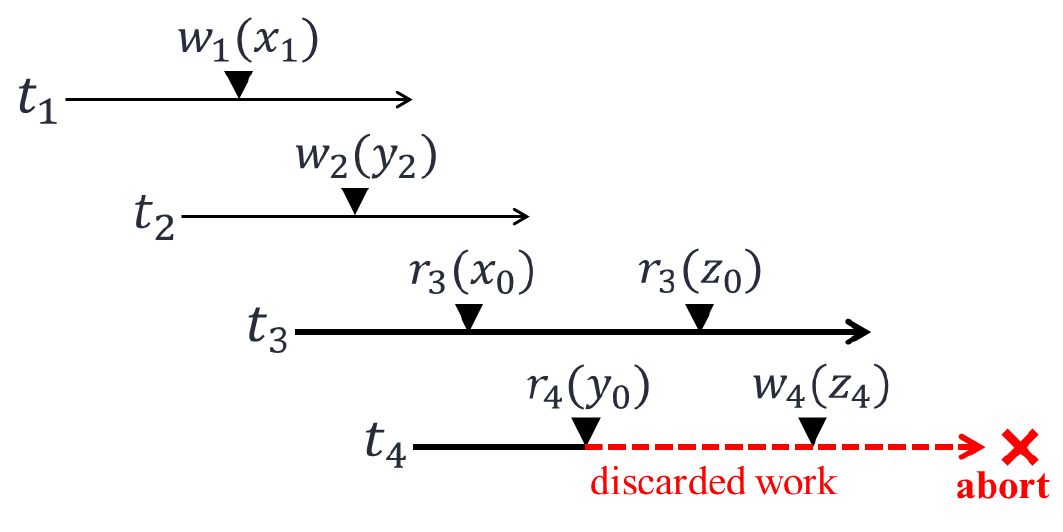}
  \caption{
    Schedule $M_1$ with long-lived $t_3$ and $t_4$.
    SSN aborts $t_4$ due to a back-edge to $t_2$ and a forward-edge from $t_3$,
    whereas ESSN relaxes the exclusion condition via forward propagation and commits $t_4$.
  }
  \label{fig:schedule-m1}
\end{figure}

\paragraph{Invariant-driven algorithm.}
We implement the ESSN exclusion condition as a single commit-time check. Enforcing it yields a structural
invariant on each version chain---\emph{strict monotonicity of $\pi$}---so the check updates only the predecessor
version and never scans the chain. The resulting protocol is DSG-based and previous-edge-only, with no early aborts.
Its commit-time work is linear in \( |R|{+}|W| \) and it matches SSN's per-version footprint while strictly admitting 
more schedules. Our formulation builds on Adya's DSG framework~\cite{adya1999-weakconsistency} and is consistent
with practical serializability work such as PSSI~\cite{revilak2011-pssi}. We also discuss the
\emph{begin-ordered} KTO: the same invariant applies without structural changes, and we analyze
its operational implications---in particular, the potential for \emph{commit-stalls}---and show that
ESSN confines them to local version-chain state.

\paragraph{Long transactions as a first-class target.}
We make \emph{mixed} workloads explicit by defining a \emph{Long} transaction via \emph{strict interval containment}:
a transaction $t$ is \emph{long} if there exists a short transaction $s$ such that 
$b(t) \prec b(s) \prec c(s) \prec c(t)$;

\emph{Long-RW} denotes a Long that performs at least one write--distinct from Long-RO handled
by safe-snapshot techniques~\cite{ports2012-ssi,shioi2024-rss}.
Under strict interval containment, the direction of rw dependencies is fixed by the chosen KTO
(commit- vs.\ begin-ordered), which makes back-edge--averse criteria conservative by construction.
We formalize this workload class and \emph{quantify} the SSN--ESSN acceptance gap experimentally.

Finally, ESSN aligns with system trends:
platforms like \emph{Google Spanner}~\cite{corbett2013-spanner-tocs} and
\emph{Amazon Redshift}~\cite{armenatzoglou2022-redshift} move toward serializability by default,
and \emph{K2}~\cite{song2025-k2} revisits MVTO with tightly synchronized clocks.
As hybrid/multiversion CC strategies proliferate~\cite{hong2025-hdcc,kambayashi2023-shirakami,nemoto2025-oze,Vandevoort2025-mvsplit}, a criterion that treats VF, VO, and KTO as separate components 
while remaining lightweight becomes increasingly valuable.
Recent developments further highlight ESSN's relevance:
system-level schedulers such as MVSchedO~\cite{cheng2024-towards}
have begun to operationalize begin-ordered alignment through per-key 
commit-stalls and stall-bypass mechanisms, while static frameworks such as 
MV-Split~\cite{Vandevoort2025-mvsplit} classify the safety of multiversion 
protocols based on their version-function and version-order alignment.
In parallel, decentralized MVSG controllers such as Oze~\cite{nemoto2025-oze}
employ order-forwarding to admit mixed long/short workloads more aggressively.
ESSN complements these directions by providing a unified runtime criterion 
that bridges static reasoning and dynamic enforcement of serializability.

\paragraph{Contributions.}
\begin{itemize}
  \item \textbf{ESSN: criterion + exclusion test.}
        We (i) factor serializability into $(\mathrm{VF},\mathrm{VO},\mathrm{KTO})$ and state a criterion requiring 
        VF/VO alignment with the declared KTO, and (ii) give an MVSG-based commit-time exclusion test valid for any 
        KTO. We formally prove correctness and that ESSN strictly subsumes SSN.
  \item \textbf{Invariant-driven algorithm with SSN-like cost.}
        With a \emph{single commit-time check}---independent of commit order---ESSN enforces monotonicity along 
        per-item version chains (no chain traversal), yielding a previous-edge-only protocol with no early aborts 
        and $\mathcal{O}(|R|{+}|W|)$ commit-time work; it \emph{strictly subsumes} SSN without increasing per-version 
        state.
  \item \textbf{Long/Long-RW as a first-class target.}
        We formalize Long/Long-RW via interval containment and explain how, under OCC/MVTO, 
        the chosen KTO fixes the direction of rw dependencies, causing conventional criteria to be conservative;
        experiments (commit-ordered KTO) show ESSN admits more schedules than SSN, with clear gains on Long-RW.

\end{itemize}

\section{Preliminaries and Background}

This section introduces the theoretical concepts essential to understanding the Extended Serial Safety Net
(ESSN) protocol. We begin by clarifying the role of the MVSG and VO in determining 
serializability~\cite{bernstein1983-mvcc,weikum2001-tis}.
We then introduce the KTO as a declared total order over transactions and make explicit
the \emph{alignment of VF/VO with KTO}, which governs how dependencies are classified.

\subsection{MVSG and VO: Defining Serializability via Dependencies}

The MVSG is a directed graph that represents the dependency relationships among transactions in a 
multiversion schedule $m$. Two components determine its structure:

1. A \textbf{VF}, which specifies for each read $r_i(x_j)$ the version $x_j$ that is returned. 
   This immediately fixes all \textit{read-from} edges (wr-dependencies) and thereby forms the schedule $m$.  

2. A \textbf{VO} $\ll_v$, a total order over all versions in $m$. 
Given the read-from edges established by the VF, the VO determines the direction of all $ww$ and $rw$ dependencies. 
Together, these two choices fully determine the MVSG$(m, \ll_v)$. 

\paragraph{MVSG construction rules (Dependency Types).}  
Given a multiversion schedule $m$ and a \textbf{VO} $\ll_v$, the MVSG$(m, \ll_v)$ is constructed as follows. 
Each node corresponds to a transaction, and directed edges are added based on the following dependency types:

\begin{itemize}
  \item \textbf{wr-dependency (read-after-write)}: if $r_k(x_j)$ exists in $m$, then add an edge $t_j \rightarrow t_k$.
  \item \textbf{ww-dependency (write-after-write)}: if $r_k(x_j)$ and $w_i(x_i)$ exist in $m$ and $x_i \ll_v x_j$, then add an edge $t_i \rightarrow t_j$.
  \item \textbf{rw-dependency (write-after-read), also known as \textit{anti-dependency}}: if $r_k(x_j)$ and $w_i(x_i)$ exist in $m$ and $x_j \ll_v x_i$, then add an edge $t_k \rightarrow t_i$.
\end{itemize}
Note that the type of an edge can vary with the chosen VO; once the VO is fixed, it determines a single 
MVSG for a given schedule.

\subsection{KTO: Structuring the Serialization Reference}

The KTO, denoted~$\ll_K$, is an explicitly defined total order over transactions. It serves as the reference order for reasoning about dependencies and serialization.

In ESSN, each edge $t_i \rightarrow t_j$ in the MVSG is classified with respect to the KTO:
\begin{itemize}
  \item If $t_i \ll_K t_j$, the edge is a \textbf{forward-edge} , denoted $t_i \xrightarrow{f} t_j$.
  \item If $t_j \ll_K t_i$, the edge is a \textbf{back-edge}, denoted $t_i \xrightarrow{b} t_j$.
\end{itemize}
 
Once the MVSG is formed and edges are annotated as $b$ or $f$
with respect to KTO, cycle detection is purely graph-theoretic; 
it no longer depends on how the MVSG was constructed or which KTO variant was chosen.
This means that an exclusion condition formulated for commit-ordered CC can also be applied to begin-ordered CC.

\paragraph{VF/VO Alignment with KTO.}

We formally define alignment of VF and VO with respect to a given KTO $\ll_K$ as follows.
Let $\mathsf{VF}$ be a version function such that for each read $r_i(x_j)$, $\mathsf{VF}(r_i) = x_j$ and 
$x_j$ was written by transaction $t_j$. Then $\mathsf{VF}$ is \emph{aligned} with $\ll_K$ if:
\[
t_j \ll_K t_i
\]
That is, the writer of the read version precedes the reader in KTO.
Let $\ll_v$ be a VO over versions. Then $\ll_v$ is \emph{aligned} with $\ll_K$ 
if for every pair of versions $x_i \ll_v x_j$ written by $t_i$ and $t_j$, respectively, we have:
\[
t_i \ll_K t_j
\]
Once a KTO and aligned VF/VO are fixed, all wr and ww edges in the MVSG$(m, \ll_v)$ are forward-edges 
with respect to $\ll_K$. Therefore, the only possible source of back-edges is rw-dependencies (anti-dependencies).

\medskip
\noindent\textit{Examples.}
\begin{enumerate}
  \item \textbf{Writes only.}
  \[
    w_0(x_0)\; w_1(x_1)\; c_0\; c_1\ .
  \]
  Even with VO declared, this schedule has no read $r_i(x_j)$, so the MVSG construction rules do not fire:
  no edges appear in the MVSG.

  \item \textbf{Add a read by $t_2$ from $t_1$.}
  \[
    w_0(x_0)\; w_1(x_1)\; c_0\; c_1\; r_2(x_1)\; c_2\ .
  \]
  With a VF aligned with a commit-ordered KTO, we obtain a $\mathrm{wr}$ edge $t_1 \to t_2$, which is \emph{forward}.
  If the VO is also aligned with the commit-ordered KTO, then $x_0 \ll_v x_1$; by the MVSG rules this yields a
  $\mathrm{ww}$ edge $t_0 \to t_1$, also \emph{forward}.

  \item \textbf{Swap $c_0$ and $c_1$; choose $x_1 \ll_v x_0$ (VO aligned with the new commit order).}
  \[
    w_0(x_0)\; w_1(x_1)\; c_1\; c_0\; r_2(x_1)\; c_2\ .
  \]
  Given $x_1 \ll_v x_0$, the MVSG rule selects the \emph{next writer} of the read version $x_1$, namely $t_0$,
  hence an $\mathrm{rw}$ edge $t_2 \to t_0$. Because $c_0 \prec c_2$, this edge is \emph{back} under the
  commit-ordered KTO. Note that for $\mathrm{ww}$, the VO order on versions coincides with the edge orientation
  between their writers (e.g., $x_0 \ll_v x_1 \Rightarrow t_0 \to t_1$). By contrast, for $\mathrm{rw}$, the VO
  determines the \emph{type} (rw) and the \emph{next-writer} target; the orientation is from the \emph{reader}
  to that next writer ($t_2 \to t_0$), which need not match the VO relation between the writers of the read version
  ($t_1$ and $t_0$).

  \item \textbf{Move $c_0$ after $c_2$.}
  \[
    w_0(x_0)\; w_1(x_1)\; c_1\; r_2(x_1)\; c_2\; c_0\ .
  \]
  The same $\mathrm{rw}$ edge $t_2 \to t_0$ persists (type and orientation unchanged), but since $c_2 \prec c_0$
  it is now labeled \emph{forward}.
\end{enumerate}

\noindent
Thus, the VF fixes the $\mathrm{wr}$ \emph{edge type and its orientation}, while the VO fixes the non-$\mathrm{wr}$
(\,$\mathrm{ww}$ and selection/targeting for $\mathrm{rw}$\,) \emph{edge type and orientation}.
The KTO merely labels each already oriented edge as forward or back. In particular, when the KTO is aligned with the
VF and VO, all $\mathrm{wr}$ and $\mathrm{ww}$ edges are forward; for $\mathrm{rw}$, once the VO selects the next writer,
the edge orientation (reader $\to$ next-writer) is independent of the KTO---only its forward/back label depends on 
the KTO.

\section{Serial Safety Net (SSN) and the Motivation for ESSN}
\label{sec:ssn-limitations}

In this section, we introduce the SSN protocol, 
a lightweight serializability checker widely adopted in multiversion concurrency control (MVCC) systems. 
We first restate its exclusion condition and analyze its implications under our criterion, 
then illustrate its limitations through a concrete example that motivates the generalization introduced by ESSN. 
Finally, we show that SSN strictly subsumes Serializable Snapshot Isolation (SSI) in its detection power.

\subsection{SSN: A Low-Overhead Commit-Time Checker}
SSN assumes that each record maintains a totally ordered sequence of versions; every write appends a new version 
to the end of that sequence. This differs slightly from our MVSG formulation but is consistent with our DSG-style 
implementation---more restrictive, yet sound.

SSN further assumes an underlying CC that forbids dirty reads and lost writes (``dirty writes''). 
These are typically enforced by recovery properties such as \emph{avoiding cascading aborts} (ACA) and \emph{strict}. 
Because ACA implies a commit-ordered VF, our alignment criterion then suggests a commit-ordered 
VO and a commit-ordered KTO. This matches SSN's claim that the only back-edges arise from 
\emph{rw-dependencies}. Although SSN does not explicitly specify a VO, under our framework it can be viewed 
as using a VF and VO aligned with a commit-ordered KTO. 
The SSN paper explicitly mentions RC and SI as 
underlying isolation levels; we therefore include RC and SI in our experimental comparison between SSN and ESSN.

For clarity, the recovery properties referenced here satisfy the usual hierarchy:
\emph{strict} $\subset$ \emph{ACA} $\subset$ \emph{recoverable} (proper inclusions).

Once the $\mathrm{MVSG}(m,\ll_v)$ is fixed, SSN defines for each transaction $t$:
\begin{align}
\pi(t) &\coloneqq \min\bigl(\{\,c(u)\mid t \xrightarrow{b*} u\,\}\cup\{c(t)\}\bigr),\\
\eta(t) &\coloneqq \max\bigl(\{\,c(u)\mid u \xrightarrow{f} t\,\}\cup\{-\infty\}\bigr) \label{eq:ssn-condition}
\end{align}
and declares $t$ \emph{unsafe} (to be aborted) if
\begin{equation}
\pi(t)\le \eta(t). \label{eq:ssn-exclusion}
\end{equation}
Here $c(u)$ denotes the commit timestamp of $u$; $\xrightarrow{b*}$ is the transitive closure of back-edges; 
and $\xrightarrow{f}$ denotes a single forward-edge with respect to KTO.

\subsection{Illustrative Example: Over-Abort in a Serializable Schedule}
\label{sec:ssn-illustrative}
Consider the following multiversion schedule $M_1$, in which four transactions access and modify three items.
This example corresponds to the same schedule presented earlier in Fig.~\ref{fig:schedule-m1} of
\S\ref{sec:introduction}. Throughout the paper, we use $t_0$ to denote the initial transaction
that creates the base versions (e.g., $x_0,y_0,z_0$). For brevity $t_0$ does not appear in the schedules.

\begin{equation}
\label{eq:M1}
M_1 = w_1(x_1)\ w_2(y_2)\ r_3(x_0)\ c_1\ r_4(y_0)\ c_2\ r_3(z_0)\ c_3\ w_4(z_4)\ a_4
\end{equation}

Assume the VO $\ll_v$ is based on commit order:
\begin{equation}
x_0 \ll_v x_1,\quad y_0 \ll_v y_2,\quad z_0 \ll_v z_4
\end{equation}

The MVSG corresponding to $M_1$ is shown in Fig.~\ref{fig:mvsg-cycle}. In this graph, transactions are positioned 
from left to right according to their commit-time order (KTO). Rightward edges represent forward-edges, 
and leftward edges indicate back-edges.

\begin{figure}[t]
  \centering
  \includegraphics[width=0.7\linewidth]{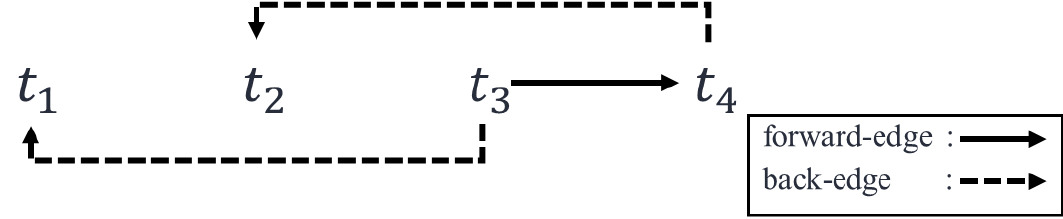}
\caption{
MVSG for schedule $M_1$. Transaction $t_2$ generates a back-edge from $t_4$ through anti-dependencies (rw-edges), 
while $t_3$ creates a forward-edge. SSN aborts $t_4$ due to the exclusion condition, although 
the overall MVSG is acyclic. ESSN, in contrast, allows $t_4$ to commit by propagating metadata through the forward-edge.
}
  \label{fig:mvsg-cycle}
\end{figure}

\paragraph{SSN Evaluation.}
We compute:
\[
\pi(t_4) = c(t_2) ,\quad \eta(t_4) = c(t_3).
\]

If $c(t_2) \leq c(t_3)$, then Eq.~\eqref{eq:ssn-exclusion} holds and $t_4$ is aborted---despite the MVSG being acyclic. 
This highlights a key weakness of SSN: it may reject valid schedules due to a conservative exclusion condition.

\subsection{SSN is more permissive than SSI under SI}\label{subsec:ssn-subsume-ssi}

We position SSN against SSI under SI.
Let $b(t)$ and $c(t)$ denote the begin and commit times of transaction $t$.

\paragraph{Concurrency under SI.}
Two transactions $t_i$ and $t_j$ are \emph{concurrent} 
iff one begins before the other commits and vice versa:
\[
b(t_i) < c(t_j) \;\land\; b(t_j) < c(t_i).
\]

Under SI, only \emph{rw}-antidependencies can occur between concurrent transactions.
In particular, concurrent \emph{wr} and \emph{ww} conflicts are disallowed:
wr is impossible since a reader only sees committed versions, 
and ww is prevented by the first-committer-wins rule.

\begin{lemma}[SSI dangerous structure]\label{lem:ssi-ds}
Fekete et al.~\cite{fekete2005-ssi} showed that every cycle in the
serialization graph under SI contains a sequence of two
\emph{rw}-antidependencies
\[
t_1 \xrightarrow{rw} t_2 \xrightarrow{rw} t_3,
\]
where $t_3$ commits before both $t_1$ and $t_2$.
Furthermore, as observed by Ports and Grittner~\cite{ports2012-ssi},
the proof of Fekete et al.\ in fact implies that $t_3$ is the
\emph{first transaction in the entire cycle to commit}.
\end{lemma}

\begin{lemma}[SSN requires two \emph{rw} edges]\label{lem:ssn-two-rw}
If the SSN exclusion condition $\pi(t)\!\le\!\eta(t)$ is violated for some transaction $t$,
then there exist at least two consecutive \emph{rw} edges contributing to this condition.
\end{lemma}

\begin{proof}
Assume there is a (shortest) dependency chain $t_1 \to t_2 \to t_3$ that triggers the SSN exclusion
$\pi(t_2)\le \eta(t_2)$. By definition of $\pi,\eta$ under SI, one of these two edges must be an
\emph{rw} back-edge to set $\pi(t_2)$, and the other must be a forward edge to set $\eta(t_2)$.
Thus we may write
\[
t_1 \xrightarrow{f} t_2 \xrightarrow{rw(b)} t_3,
\qquad
\pi(t_2)=c(t_3),\ \ \eta(t_2)=c(t_1).
\]

We now show that the forward edge $t_1\to t_2$ cannot be $\mathit{wr}$ or $\mathit{ww}$ under SI.
If $t_1\to t_2$ were $\mathit{wr}$ or $\mathit{ww}$, then by SI's non-overlap rule we have
$c(t_1)<b(t_2)$. On the other hand, since $t_2 \xrightarrow{rw} t_3$, we must have $b(t_2)<c(t_3)$
(otherwise $t_2$ would have read $t_3$’s version). Hence
\[
\eta(t_2)=c(t_1)\;<\;b(t_2)\;<\;c(t_3)=\pi(t_2),
\]
which contradicts the exclusion condition $\pi(t_2)\le \eta(t_2)$.

Therefore, the forward edge contributing to $\eta(t_2)$ must itself be an $\mathit{rw}$ edge.
In particular, to make $\pi(t_2)\le \eta(t_2)$ feasible, the chain contributing to $\pi$
and the incoming chain contributing to $\eta$ must together include
two $\mathit{rw}$ edges with the forward-in edge at $t_2$;
otherwise, $\pi(t_2)\le\eta(t_2)$ cannot hold under SI.
\end{proof}

\begin{lemma}[The last node commits first]\label{lem:ssn-t3-commit}
If two consecutive \emph{rw} edges contribute to the SSN exclusion condition,
then the last transaction commits first.
\end{lemma}

\begin{proof}
If both \emph{rw} edges are back edges, then by definition the last one commits first.
If one \emph{rw} edge is forward-in and the other is back-out, we have
$t_1 \xrightarrow{rw(f)} t_2 \xrightarrow{rw(b)} t_3$.
Here $\pi(t_2)=c(t_3)$ and $\eta(t_2)=c(t_1)$.
The exclusion condition requires $\pi(t_2)\le \eta(t_2)$,
i.e.\ $c(t_3)\le c(t_1)$.
Thus $t_3$ must be the first committer.
\end{proof}

\begin{lemma}[SSN subsumes SSI]\label{lem:ssn-subsume-ssi}
If a transaction aborts under SSN, then it also aborts under SSI.
\end{lemma}

\begin{proof}
Lemmas~\ref{lem:ssn-two-rw} and~\ref{lem:ssn-t3-commit} show that
SSN exclusion requires two consecutive \emph{rw} edges,
with the last committer first to commit.
This matches exactly the SSI dangerous structure,
so any SSN abort implies an SSI abort.
\end{proof}

\begin{remark}[The converse does not hold]
There exist SI histories with an all-back-edge chain
$t_1 \xrightarrow{rw(b)} t_2 \xrightarrow{rw(b)} t_3$
where $c(t_3)<c(t_2)<c(t_1)$.
Here SSI aborts because the dangerous structure is present,
while SSN does not, since there is no forward-in edge to constrain $\eta(\cdot)$.
\end{remark}

\noindent\textbf{Conclusion.} Under SI, $A_{\mathrm{SSN}} \supset A_{\mathrm{SSI}}$ 
(i.e., SSN is strictly more permissive than SSI).

\smallskip
\noindent
Following the SSI literature~\cite{cahill2009-ssi-tods}, 
we use the term \emph{pivot} to denote the middle transaction in the classic
three-node dangerous structure 
$t_1 \xrightarrow{f} t_2 \xrightarrow{b} t_3$,
where the rightmost transaction commits first in the cycle.
SSN inherits this notion and detects such pivot structures to prevent
serialization cycles in the MVSG, 
thereby achieving its strict subsumption of SSI.

\section{Extended Serial Safety Net (ESSN)}
\label{sec:essn}
This section introduces ESSN, an extension of SSN that generalizes the definition of transaction ordering and 
exclusion criteria to broaden the space of histories that can be accepted as serializable. 
We first generalize the notion of KTO, then redefine the exclusion condition 
using $\pi(t)$ and a new value $\xi(t)$. We demonstrate how ESSN eliminates false positives caused by SSN 
and prove its correctness by showing that any cycle in the MVSG will necessarily include at least one abort target 
under ESSN.

\subsection{Generalizing the KTO}

In SSN, the definitions of the KTO, as well as the values $\pi(t)$ and $\eta(t)$, are based on the 
commit timestamp $c(t)$ of transaction $t$. ESSN generalizes this by replacing $c(t)$ with an arbitrary function 
$\sigma(t)$ that returns a totally ordered value.

For example, in MVTO, a well-known begin-ordered CC 
(see Tab.~\ref{tab:cc_comparison} in Section~\ref{sec:related-work} for details), 
$\sigma(t)$ may represent the begin timestamp of transaction~$t$. 
With this generalization, algorithms similar to SSN can be defined even under different orderings.

Since ESSN evaluates its exclusion condition at commit-time, the relevant metadata
($\pi(t)$ and $\xi(t)$ introduced below) must be finalized when $t$ commits.
If the KTO is not commit-based, this may require delaying $t$'s commit
until the outcome of other transactions is known.
We refer to this phenomenon as a \emph{commit-stall}, and analyze its
implications in detail in \S\ref{subsec:begin-stalls}.

\subsection{Definition of the ESSN exclusion test}

ESSN extends the SSN exclusion test for general KTOs. Like SSN, ESSN decides 
at commit time whether transaction $t$ would introduce a cycle in the MVSG$(m, \ll_v)$.
If a cycle may be formed, $t$ is aborted; otherwise, it is allowed to commit. This ensures that $m$ remains MVSR.

The KTO $\ll_K$ is generalized using a function taking values in a totally ordered domain $\sigma(t)$ over transactions. Specifically:

\begin{equation}
\sigma(t_i) < \sigma(t_j) \iff t_i \ll_K t_j
\end{equation}

To determine whether $t$ may create a cycle, ESSN defines two values for each transaction $t$:

\begin{align}
\pi(t) &= \min \left( \{ \sigma(u) \mid t \xrightarrow{b*} u \} \cup \{ \sigma(t) \} \right)\notag \\
       &= \min \left( \{ \pi(u) \mid t \xrightarrow{b} u \} \cup \{ \sigma(t) \} \right) \\
\xi(t) &= \max \left( \{ \pi(u) \mid u \xrightarrow{f} t \} \cup \{-\infty\} \right) \label{eq:essn-condition}
\end{align}

While $\eta(t)$ in SSN directly uses commit times of forward predecessors, ESSN's $\xi(t)$ incorporates the $\pi$ values of these predecessors, allowing deeper propagation of serializability constraints. This results in fewer unnecessary aborts.

The ESSN exclusion condition is as follows:

\begin{equation}
\pi(t) \leq \xi(t) \label{eq:essn-criterion}
\end{equation}

If Eq.~\eqref{eq:essn-criterion} holds, transaction $t$ is aborted; otherwise, it is allowed to commit.

\subsection{Relationship Between ESSN and SSN}

ESSN is a generalization of SSN that admits a strictly larger class of serializable histories. This section 
formally establishes their relationship.

\subsubsection{Inclusion Relationship}

The following two lemmas clarify how the ESSN exclusion criterion is weaker (i.e., more permissive) than that of SSN.

\begin{lemma} \label{lem:xi-le-eta}
Let $\sigma(t)$ be the commit timestamp $c(t)$. Then, $\xi(t) \leq \eta(t)$.
\end{lemma}

\begin{proof}
If $\xi(t) = -\infty$, the inequality clearly holds. Otherwise, by definition (Eq.~\eqref{eq:essn-condition}), 
there exists a transaction $u$ such that $u \xrightarrow{f} t$ and $\xi(t) = \pi(u)$.  
From the definition of $\eta(t)$ (Eq.~\eqref{eq:ssn-condition}), it follows that $c(u) \leq \eta(t)$.  
Since $\pi(u) \leq \sigma(u)$ by definition, and $\sigma(u) = c(u)$ in SSN, we get:

\[
\xi(t) = \pi(u) \leq \sigma(u) = c(u) \leq \eta(t)
\]

Hence, $\xi(t) \leq \eta(t)$.
\end{proof}

\begin{lemma}
There exist histories where SSN would abort a transaction, but ESSN would not.
\end{lemma}

\begin{proof}
Compare the SSN and ESSN exclusion conditions:

\begin{align}
\text{SSN: } \pi(t) &\leq \eta(t) \tag{\ref{eq:ssn-exclusion}} \\
\text{ESSN: } \pi(t) &\leq \xi(t) \tag{\ref{eq:essn-criterion}}
\end{align}

By \Cref{lem:xi-le-eta}, we have $\xi(t) \leq \eta(t)$. Therefore:

\[
\pi(t) \leq \xi(t) \leq \eta(t)
\]

This implies that SSN may abort a transaction even if the ESSN condition is satisfied. 
Hence, there are cases where SSN aborts but ESSN allows commit.
\end{proof}

\paragraph{Example 2 (Revisiting Example~1).}
We revisit the same schedule $M_1$ introduced in Example~1, but now evaluate it using ESSN's exclusion condition.
Consider the MVSG shown in Fig.~\ref{fig:essn-forward-avoid}.
In this graph, there is a back-edge $t_4 \xrightarrow{b} t_2$, which would trigger the SSN exclusion condition
$\pi(t_4)\le \eta(t_4)$ (with $\sigma(\cdot)=c(\cdot)$).

However, $t_3 \xrightarrow{f} t_4$ is a forward-edge. Under ESSN, this edge
\emph{sets} the forward-side bound of $t_4$ by propagating the forward predecessor's $\pi$:
\[
  \xi(t_4) \;\leftarrow\; \max\bigl\{\xi(t_4),\, \pi(t_3)\bigr\}.
\]
Consequently, $t_4$ is allowed to commit because $\pi(t_4)>\xi(t_4)$, even in cases
where SSN would abort due to $\pi(t_4)\le \eta(t_4)$.
This example illustrates how ESSN reduces unnecessary aborts through more precise forward-side propagation.

\begin{figure}[t]
  \centering
  \includegraphics[width=0.7\linewidth]{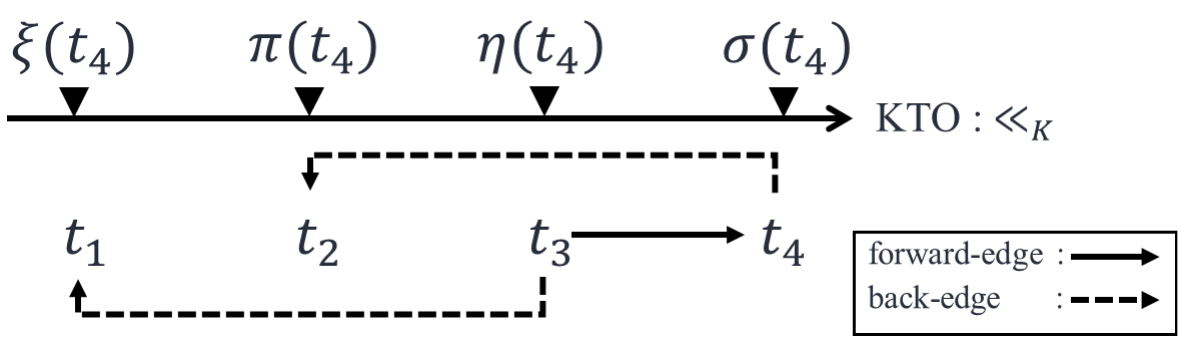}
\caption{MVSG for Example~2. SSN aborts $t_4$ due to $\pi(t_4)\le \eta(t_4)$,
whereas ESSN \emph{sets} $\xi(t_4)$ to $\pi(t_3)$ via the forward-edge $t_3\xrightarrow{f}t_4$
and allows $t_4$ to commit when $\pi(t_4)>\xi(t_4)$.}

\label{fig:essn-forward-avoid}

\end{figure}

\subsubsection{Abort Target and Cycle Resolution}

While ESSN generalizes SSN, it does not necessarily designate the same abort target(s) when a cycle is present 
in the MVSG.

Consider a variant of the MVSG from Fig.~\ref{fig:essn-forward-avoid}, where we add a forward-edge $t_2 \rightarrow t_3$ 
(see Fig.~\ref{fig:mvsg-with-extra-edge}). 
The transactions are still arranged left to right according to a commit-time KTO.

\begin{itemize}
  \item Under SSN, both $t_3$ and $t_4$ are abort targets.
  \item Under ESSN, only $t_3$ is an abort target, since
    \[
      \pi(t_3) = \sigma(t_1), \quad \xi(t_3) = \sigma(t_2),
    \]
    and $\pi(t_3) \le \xi(t_3)$.
  \item $t_4$ is not an abort target under ESSN, by the same reasoning as in Fig.~\ref{fig:essn-forward-avoid}.
\end{itemize}

As shown in Fig.~\ref{fig:mvsg-with-extra-edge}, the forward-edge $t_2 \rightarrow t_3$ contributes directly 
to $\xi(t_3)$. Intuitively, if $\sigma(t_1) < \sigma(t_2)$ and there is a forward path to $t_3$, then $t_3$ should 
be responsible for breaking the cycle; aborting $t_3$ renders aborting $t_4$ unnecessary. 
If instead $\sigma(t_2) < \sigma(t_1)$, the ESSN condition may trigger for $t_4$.

\begin{figure}[t]
  \centering
  \includegraphics[width=0.7\linewidth]{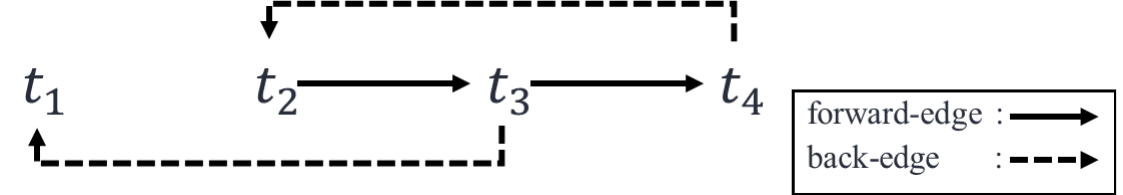}
  \caption{MVSG with an additional forward-edge from $t_2$ to $t_3$.}
  \label{fig:mvsg-with-extra-edge}
\end{figure}

\paragraph{Comparing ``chain vs.\ chain''.}
Our interest in ESSN also originated from a case (Fig.~\ref{fig:mvsg-with-extra-edge}) where SSN marked
\emph{two} transactions ($t_3$ and $t_4$) as abort candidates after adding a 
``whisker'' \emph{back-edge} $t_3 \rightarrow t_1$.
This reflects an asymmetry in SSN's exclusion test: it compares a \emph{back-edge chain} against a \emph{single} 
forward predecessor.
By contrast, adding the whisker back-edge yields a more symmetric comparison---aggregated chains on both sides.
This symmetry also makes implementation simpler.

\subsection{Correctness of the ESSN Criterion}

We now prove that any transaction allowed to commit under the ESSN exclusion condition preserves MVSR.

\begin{theorem}
All transactions that successfully commit under ESSN produce an MVSG that is acyclic, i.e., ESSN guarantees MVSR.
\end{theorem}

\begin{proof}
Suppose, for contradiction, that MVSG$(m, \ll_v)$ contains a cycle $C_n$ with $n \ge 2$ edges.
Let $t_n$ be the transaction in $C_n$ with the smallest $\sigma(t)$.
Then we can write the cycle (with cyclic indices modulo $n$) as
\[
t_1 \rightarrow t_2 \rightarrow \cdots \rightarrow t_{n-1} \bedge t_n \fedge t_1 .
\]
In particular, $t_{n-1} \bedge t_n$ and $t_n \fedge t_1$.

Define the set of forward-edge sources by
\begin{equation}
F := \{\, t_i \mid t_i \xrightarrow{f} t_{i+1} \,\}, \qquad \text{where } t_{n+1} := t_1.
\end{equation}
By the minimality of $\sigma(t_n)$ we have $t_n \fedge t_1$, hence $t_n \in F$ and so $F \neq \varnothing$.

Let $t_k \in F$ be a transaction that maximizes $\pi(\cdot)$ over $F$. Then $t_k \xrightarrow{f} t_{k+1}$ and by definition:

\[
\pi(t_k) \leq \xi(t_{k+1})
\]

We show that $\pi(t_{k+1}) \leq \pi(t_k)$, which implies:

\[
\pi(t_{k+1}) \leq \xi(t_{k+1})
\]

and thus $t_{k+1}$ is an abort target under ESSN.

\textbf{Case 1:} If $t_{k+1} \xrightarrow{f} t_{k+2}$, then $t_{k+1} \in F$. By definition of $t_k$, $\pi(t_{k+1}) \leq \pi(t_k)$.

\textbf{Case 2:} If $t_{k+1} \xrightarrow{b} t_{k+2}$, let $t_j$ be the transaction reachable from $t_{k+1}$ via back-edges such that $\sigma(t_j)$ is minimized. Then $t_j \xrightarrow{f} t_{j+1}$, so $t_j \in F$. By definition of $t_k$:

\[
\pi(t_j) \leq \pi(t_k)
\Rightarrow \pi(t_{k+1}) \leq \pi(t_j) \leq \pi(t_k)
\]

Hence:

\[
\pi(t_{k+1}) \leq \xi(t_{k+1})
\]

and $t_{k+1}$ is an abort target under ESSN.

\end{proof}

\noindent
\textbf{Summary:} Any transaction that commits under ESSN cannot participate in a cycle, and hence all committed transactions form an acyclic MVSG. Therefore, ESSN guarantees serializability.

\section{Implementation of ESSN}
\label{sec:implementation}

This section presents a commit-time implementation of ESSN for a typical MVCC engine.
We first outline the metadata layout and then describe the commit protocol. As in SSN, we maintain
per-version metadata along the version chain and exploit \emph{previous-edge-only} maintenance, so that
writers inspect only their immediate predecessor at commit time. This avoids chain scans and
keeps version objects compact. No early aborts are performed at read or write time.
We then discuss the operational implications of a begin-ordered KTO under ESSN 
and show how the \emph{anti-pivot} structure illustrates ESSN's structural advantage,
extending naturally to general choices of KTO.

\paragraph{Reduction to adjacent edges (direct dependencies).}
We use a conventional \emph{adjacent-reduced} MVSG: for each key, let its versions be
ordered by the declared VO $\ll_v$ as $v_0 \prec v_1 \prec \cdots$ with
writers $W_0,W_1,\dots$.
We keep only (i) one \emph{wr} edge $W_k \!\to\! R$ for each read $R$ that returns $v_k$,
(ii) \emph{ww} edges between \emph{adjacent} writers $W_k \!\to\! W_{k+1}$, and
(iii) for each read of $v_k$, a single \emph{rw} (anti-)dependency $R \!\to\! W_{k+1}$
to the \emph{next writer} that overwrites $v_k$.
All non-adjacent $ww$ and farther $rw$ edges are implied by the $ww$ chain (i.e., the version chain), so reachability---and thus cycle existence---is preserved.
This follows Adya's direct-dependency convention for DSG (and for MVCC, MVSG)~\cite{adya1999-weakconsistency}, 
and Revilak's subsequent proof that DSG suffices for serializability checking~\cite{revilak2011-pssi}.
\emph{Consequently, the per-version metadata can be kept minimal.}

\paragraph{Computing $\pi$ and $\xi$.}

\begin{lemma}[Strict monotonicity of $\pi$ on a version chain]
\label{lem:pi-monotone}
Let a key's versions be ordered by the declared VO $\ll_v$ and aligned with the KTO
(e.g., commit-ordered $\ll_v$ with commit-ordered KTO). For adjacent writers
$W_k \xrightarrow{ww} W_{k+1}$ that both commit under ESSN, we have
$\pi(W_k) < \pi(W_{k+1})$.
\end{lemma}

\begin{proof}
Because $W_k \xrightarrow{ww} W_{k+1}$ is a forward-edge (VO aligned with KTO), by definition
\[
  \xi(W_{k+1})
  = \max\{\;\pi(u) \mid u \xrightarrow{f} W_{k+1}\;\}
  \implies
  \xi(W_{k+1}) \ge \pi(W_k).
\]
If, to the contrary, $\pi(W_k) \ge \pi(W_{k+1})$, then $\xi(W_{k+1}) \ge \pi(W_{k+1})$, so the
exclusion condition $\pi \le \xi$ holds and $W_{k+1}$ would be aborted—contradiction.
Hence $\pi(W_k) < \pi(W_{k+1})$.
\end{proof}

\subsection{Metadata Layout}

As in SSN, ESSN maintains two kinds of state: \emph{transaction objects} and \emph{version objects} 
(see Table~\ref{tab:essn-meta}). A transaction object records the transaction's read set and write set; 
each version object stores compact per-version metadata.

\noindent\emph{Per-version state (guided by Lemma~\ref{lem:pi-monotone}).}
The strict monotonicity of $\pi$ along each version chain lets us store just three stamps per version:

\begin{itemize}
  \item \textbf{\texttt{sstamp} (successor minimum).}
  On the back-edge side, the \emph{smallest} $\pi$ that can constrain a reader of a version $v$ is the $\pi$ of the 
  \emph{next writer} that overwrites $v$; any later overwriter has strictly larger $\pi$ by Lemma~\ref{lem:pi-monotone}.
   Hence $v.\texttt{sstamp}$ is that next writer's $\pi$ (or $+\infty$ if none), 
   and readers compute their $\pi$ as the minimum of the read versions' \texttt{sstamp}s (together with $\sigma(t)$).

  \item \textbf{\texttt{crepi} ($=\pi(\text{creator})$).}
  For $ww$ dependencies, the \emph{only} predecessor that can affect a new writer's $\xi$ is the 
  \emph{immediate predecessor}; earlier writers have smaller $\pi$ and are dominated by the monotonicity of $\pi$.
  We record the creator's $\pi$ in each version as \texttt{crepi}; a writer reads \texttt{prev.crepi} at commit.
  \emph{For $wr$ dependencies}, the writer of the read version is a single forward predecessor, so the reader directly
  incorporates $\pi(\mathrm{creator}(u))=\texttt{u.crepi}$ into $\xi$ (no chain traversal).

  \item \textbf{\texttt{psstamp} (forward prefix maximum for $rw$).}
  Forward $rw$ contributions to $\xi$ arise from \emph{readers} of earlier versions on the chain; 
  because there may be many readers, what a future writer needs is the \emph{prefix maximum} over reader 
  $\pi$ up to the predecessor link.
  We maintain on each version $v$ a prefix high-water mark
  $v.\texttt{psstamp}=\max\{\pi(r)\}$ over all readers of ancestors up to $v$; a writer reads 
  \texttt{prev.psstamp} at commit.
\end{itemize}

\begin{table}[t]
\centering
\begin{tabular}{lll}
\toprule
Field & Description & Initial value\\
\midrule
\texttt{t.status}  & in-flight / committed / aborted & in-flight \\
\texttt{t.sstamp}  & $\pi(T)$: min over back-edge side & $\infty$ \\
\texttt{t.psstamp} & $\xi(T)$: max over forward predecessors & $-\infty$ \\
\texttt{t.reads}   & Read set & $\emptyset$ \\
\texttt{t.writes}  & Write set & $\emptyset$ \\
\midrule
\texttt{v.sstamp}  & next-writer's $\pi$ (back-propagated) & $\infty$ \\
\texttt{v.psstamp} & prefix max of readers' $\pi$ up to this link & $-\infty$ \\
\texttt{v.crepi}   & creator's $\pi$ (for $wr$/$ww$ direct-in) & $-\infty$ \\
\texttt{v.prev}    & pointer to the overwritten version & \textsc{null} \\
\bottomrule
\end{tabular}
\caption{ESSN metadata. Per-version \texttt{psstamp} is maintained via the \textit{prev} link
with previous-edge-only updates, so a writer consults only its immediate predecessor at commit.
\texttt{crepi} stores the creator's $\pi$ for $wr$/$ww$ direct-in.}
\label{tab:essn-meta}
\end{table}

\subsection{ESSN Read/Write (record only)}

During execution, the system only records the versions accessed: reads are added to the
transaction's read set and writes to the write set. No exclusion check or per-version update
is performed at read/write time. All reasoning and metadata updates are deferred to commit.

\subsection{ESSN Commit Protocol}
\begin{algorithm}[t]
\caption{ESSN Commit Protocol (previous-edge-only; no commit timestamp).}
\label{alg:essn-commit}
\begin{small}
\begin{algorithmic}[1]
\Function{essn\_commit}{$t$}
  \State $\sigma \gets \Call{next\_kto\_scalar}{}$        \Comment{e.g., commit KTO}
  \State $t.\texttt{sstamp} \gets \sigma$                 \Comment{seed $\pi(t)$}
  \State $t.\texttt{psstamp} \gets -\infty$                     \Comment{init $\xi(t)$}
  \Statex
  \ForAll{$u \in t.\texttt{reads}$}                       \Comment{reads: back $\pi$ \& wr direct-in}
    \State $t.\texttt{sstamp} \gets \min(t.\texttt{sstamp},\, u.\texttt{sstamp})$
    \State $t.\texttt{psstamp} \gets \max(t.\texttt{psstamp},\, u.\texttt{crepi})$
  \EndFor
  \Statex
  \ForAll{$v \in t.\texttt{writes}$}                      \Comment{writes: ww via $p.\mathrm{crepi}$, rw via $p.\mathrm{psstamp}$}
    \State $p \gets v.\texttt{prev}$
    \If{$p \neq \textsc{null}$}
      \State $t.\texttt{psstamp} \gets \max(t.\texttt{psstamp},\, p.\texttt{crepi})$    \Comment{ww}
      \State $t.\texttt{psstamp} \gets \max(t.\texttt{psstamp},\, p.\texttt{psstamp})$  \Comment{rw}
    \EndIf
  \EndFor
  \Statex
  \If{$t.\texttt{sstamp} \le t.\texttt{psstamp}$}         \Comment{single ESSN exclusion}
    \State \Call{abort}{$t$}
    \State \Return
  \EndIf
  \State $t.\texttt{status} \gets \textsc{Committed}$
  \Statex
  \ForAll{$v \in t.\texttt{writes}$}                      \Comment{install new versions \& propagate}
    \State $p \gets v.\texttt{prev}$
    \If{$p \neq \textsc{null}$}
      \State $p.\texttt{sstamp} \gets t.\texttt{sstamp}$  \Comment{back-propagate next-writer $\pi$}
    \EndIf
    \State $v.\texttt{crepi} \gets t.\texttt{sstamp}$     \Comment{creator's $\pi$ for future wr/ww}
    \If{$p = \textsc{null}$}
      \State $v.\texttt{psstamp} \gets -\infty$
    \Else
      \State $v.\texttt{psstamp} \gets p.\texttt{psstamp}$  \Comment{carry readers' prefix}
    \EndIf
  \EndFor
  \ForAll{$u \in t.\texttt{reads}$}
    \State $u.\texttt{psstamp} \gets \max(u.\texttt{psstamp},\, t.\texttt{sstamp})$ \Comment{register reader's $\pi$}
  \EndFor
\EndFunction
\end{algorithmic}
\end{small}
\end{algorithm}

As shown in Algorithm~\ref{alg:essn-commit}, the commit protocol has two phases for a
transaction object $t$: an \emph{evaluation} phase and a \emph{finalization} phase that
installs writes and propagates the required metadata. We write the read and write sets
as $\mathrm{reads}(t)$ and $\mathrm{writes}(t)$.

\paragraph{Evaluation Phase.}
We finalize $\pi(t)$ as
$\displaystyle \pi(t)=\min\!\bigl(\,\sigma(t),\;\min_{u\in \mathrm{reads}(t)} u.\texttt{sstamp}\bigr)$,
and we finalize $\xi(t)$ as
\[
\xi(t)
=\max\!\Bigl(
  \underbrace{\max_{u\in \mathrm{reads}(t)} u.\texttt{crepi}}_{\text{$wr$ direct input}}\;,
  \underbrace{\max_{v\in \mathrm{writes}(t)} v.\texttt{prev}.\texttt{crepi}}_{\text{$ww$ direct input}}\;,
  \underbrace{\max_{v\in \mathrm{writes}(t)} v.\texttt{prev}.\texttt{psstamp}}_{\text{$rw$ prefix}}
\Bigr).
\]
We then apply the exclusion rule: commit if and only if (\emph{iff}) $\pi(t)>\xi(t)$, otherwise abort.

\paragraph{Finalization Phase.}
For each new version $v$ produced by $t$, set
$v.\texttt{crepi}\!\leftarrow\!\pi(t)$ and
$v.\texttt{psstamp}\!\leftarrow\!v.\texttt{prev}.\texttt{psstamp}$,
and back-propagate $\pi(t)$ to $v.\texttt{prev}.\texttt{sstamp}$.
At $t$'s commit, for each read version $u\in \mathrm{reads}(t)$, update
$u.\texttt{psstamp}\!\leftarrow\!\max(u.\texttt{psstamp},\,\pi(t))$.
Consequently, a committing writer consults only the immediate predecessor
$p\!=\!v.\texttt{prev}$ on each written key.

\paragraph{Overhead.}
Reads and writes only record set membership; there are no per-version metadata updates and no early aborts.
At commit time, the protocol scans the read and write sets twice---once to compute the exclusion check 
and once to propagate updates: for each written version 
it updates the predecessor's \texttt{sstamp} (back-propagating $\pi$) and carries the forward bound $\xi$, 
while for each read version it updates the per-version KTO high-water mark. No chain walk is required.
The overall complexity is $O(\lvert\mathrm{reads}\rvert+\lvert\mathrm{writes}\rvert)$---on par with SSN---while keeping per-version objects compact via previous-edge-only maintenance.

\emph{Read-time shortcut (SSN-style).} If, at read time, the next-version writer (overwriter) of the read version has 
already committed, we \emph{copy} its $\pi$ into the reader's state and may drop that entry from the read set. 
This is an optimization only and does not affect correctness.

\subsection{Operational Implications of Begin-Ordered KTO under ESSN}
\label{subsec:begin-stalls}
One concern with a begin-ordered KTO is that it may introduce a commit-stall, specifically when the
KTO diverges from the commit order. Recall that $\pi(t_i)$ and $\xi(t_i)$ in the ESSN exclusion
conditions are defined with respect to transactions that precede $t_i$ in the KTO. For back-edges, this concerns
outgoing edges from $t_i$; for forward-edges, incoming edges to $t_i$---in both cases, the counterpart must be
earlier in the KTO. With a commit-ordered KTO, those earlier transactions have already committed, so no waiting
is required. With a begin-ordered KTO, however, some of these KTO-predecessors may not have committed yet,
forcing $t_i$ to wait until they finish so that its $\pi$ and $\xi$ can be finalized.

This situation is similar to 2V2PL~\cite{weikum2001-tis}, where each commit must wait
until the completion of any reading transactions, even though 2V2PL is commit-ordered.
As noted in the 2V2PL description, conventional 2PL is costly because reads and writes
block each other, causing frequent waits---especially for reads. By contrast, under 2V2PL
a read always observes the previous version and is therefore compatible with a concurrent
write; the cost is shifted to commit time, where writers wait for the reader to finish
(i.e., a commit-stall).

\paragraph{Operational notes on commit-stall.}
Both 2PL and 2V2PL rely on \emph{item-based locking} and thus require maintaining per-item lock chains
(though the burden is lighter for 2V2PL). In contrast, the \emph{commit-stall} used here maintains no lock
chains at all: it simply delays a later transaction until earlier-started ones (in the declared KTO) have
committed, so that stamps (e.g., $\pi,\xi$) are finalized in order. Deadlock detection is unnecessary,
because every wait follows the total order of KTO.

A downside is potential \emph{wasteful waiting}: a stalled transaction may already know that no conflict
capable of producing a dangerous structure will materialize. To mitigate this, short transactions can take an
\emph{early decision} (a ``stall-bypass'') when a local check certifies safety; additionally, prior declarations
or lightweight statistical sampling that infer long-transaction characteristics can reduce over-stalling
(e.g., uniform schedule sampling and nonparametric performance bounds as in \cite{cheng2024-towards}).

\paragraph{Example and a safe stall-bypass.}
Consider
\[
\mathtt{b_1\;\; b_2\; r_2(x_0)\; c_2\;\; w_1(x_1)\; c_1}\, .
\]
Under a begin-ordered KTO with $t_1$ (long) earliest, $t_2$ is short and read-only on $x$ in this schedule.
Even if $t_1$ later writes $x_1$, the only edge introduced with respect to $t_2$ is
$\,t_2 \xrightarrow{\mathrm{rw}(\mathrm{b})} t_1\,$, while its only incoming forward edge remains the baseline
$\,t_0 \xrightarrow{\mathrm{wr}(\mathrm{f})} t_2\,$ from the initial writer $t_0$.
Hence $t_2$ can safely commit \emph{without waiting} for $t_1$. In SSN/ESSN terms, since the exclusion condition
is $\pi(t)\!\le\!\eta(t)$, the safe case is the strict reverse. If
$\eta(t_2)$ (from $t_0$) is strictly less than $\pi(t_2)$ (via $t_1$), i.e., $\eta(t_2) < \pi(t_2)$,
then $t_2$ is admitted.

\medskip
\noindent\textit{Appearance order of version objects under a begin-ordered KTO.}
When we adopt a VO aligned with a begin-ordered KTO, commits may arrive out of
that order: a transaction that began earlier can commit after others. Without coordination,
this would \emph{logically} require inserting its version object between already-installed
versions in the physical chain.

The \emph{commit-stall} avoids such mid-chain insertions. By stalling only for transactions
that began earlier (and thus precede in KTO---and hence in VO), we ensure that
ww-related writes are materialized in KTO order. Consequently, the physical version chain
can be maintained \emph{append-only}---new versions always appear at the end of the chain---
and no re-linking or reordering of version objects is required.

\subsection{Anti-pivot structure inherent in a version chain}
\label{subsec:anti-pivot-structure}

\paragraph{Observation (anti-pivot structure from a single read).}
Let a transaction \(t_i\) read version \(x_j\) on key \(x\), written by \(t_j\).
Along the version chain ordered by \(\ll_v\) we denote successors of \(x_j\) by
\(x_{s_1} \ll_v x_{s_2} \ll_v \dots\) with writers \(t_{s_1}, t_{s_2}, \dots\).
By definition of MVSG, each successor writer induces an anti-dependency
\(t_i \xrightarrow{rw} t_{s_k}\).
Given a KTO, these edges partition into
\[
S^{-} \coloneqq \{\, t_{s_k} \mid t_{s_k} \ll_K t_i \,\} \quad(\text{back-edges}),\qquad
S^{+} \coloneqq \{\, t_{s_k} \mid t_i \ll_K t_{s_k} \,\} \quad(\text{forward-edges}).
\]
If both sets are non-empty, choose \(t_p\) and \(t_q\) as the KTO-minimal elements
of \(S^{-}\) and \(S^{+}\), respectively. Then
\(t_p \xleftarrow{\,b\,} t_i \xrightarrow{\,f\,} t_q\)
forms an \emph{anti-pivot structure} localized to key \(x\).
Operationally, a single read \(r_i(x_j)\) guarantees that all future versions on the
chain produce \(rw\)-edges out of \(t_i\);
whenever at least one of them lands before \(t_i\) in KTO and another lands after,
an anti-pivot structure arises deterministically.
In effect, even a read-only (RO) transaction can participate in
an anti-pivot structure, particularly under heavy contention on the same key.

\emph{Consequence for ESSN vs.\ SSN.}
For an anti-pivot $(t_p \xleftarrow{b} t_i \xrightarrow{f} t_q)$,
SSN may abort $t_q$ when $t_i$ contributes the largest $\eta$
among all predecessors and $\pi(t_q)\!\le\!\eta(t_q)$,
driven by a back-edge unrelated to the forward path.
Under ESSN, $\xi$ aggregates the $\pi$ values of forward predecessors;
if $\xi(t_q) < \pi(t_q)$, then $t_q$ is allowed to commit.

\paragraph{Example (anti-pivot structure).}
Consider the compact schedule
\[
m_2 = b_1\,w_1(x_1)\;\; b_2\,w_2(y_2)\;\; b_3\,r_3(x_0)\;\; c_1\;\; b_4\,r_4(y_0)\,w_4(x_4)\;\; c_2\;\; c_3\;\; c_4.
\]
Let KTO be the commit order, i.e., \(KTO: t_1 \prec t_2 \prec t_3 \prec t_4\).

This schedule induces the following MVSG edges:
\begin{itemize}
  \item \(t_1 \fedge t_4\) (ww on \(x\))
  \item \(t_3 \bedge t_1\) (rw on \(x\): \(r_3(x_0)\) then \(w_1(x_1)\))
  \item \(t_3 \fedge t_4\) (rw on \(x\): \(r_3(x_0)\) then \(w_4(x_4)\))
  \item \(t_4 \bedge t_2\) (rw on \(y\): \(r_4(y_0)\) then \(w_2(y_2)\))
\end{itemize}

On key \(x\), the version chain is \(x_0 \leftarrow w_1(x_1) \leftarrow w_4(x_4)\) with a read \(r_3(x_0)\).
By the observation in this section, the single read \(r_3(x_0)\) yields \(t_3 \xrightarrow{rw} t_1\) and
\(t_3 \xrightarrow{rw} t_4\); with \(KTO: t_1 \prec t_3 \prec t_4\), these classify as
\(t_1 \xleftarrow{b} t_3 \xrightarrow{f} t_4\), i.e., an anti-pivot structure localized to key \(x\).

\emph{ESSN vs.\ SSN.}
Here \(t_4\) also has a back-edge to \(t_2\) on key \(y\).
SSN may abort \(t_4\) when \(\pi(t_4)\!\le\!\eta(t_4)\), due to the combination of its own back-edge
(\(t_4 \bedge t_2\)) and metadata propagation across the pivot.
Under ESSN, forward-only propagation yields \(\xi(t_4) < \pi(t_4)\), so \(t_4\) commits.

\subsection{KTO in general under ESSN}
\label{subsec:kto-general}

\paragraph{KTO-insensitive case ($M_1$).}
Because the begin- and commit-ordered KTO coincides on $M_1$, swapping the KTO has
no effect on edge orientations in the MVSG.
\[
M_1 \;=\; b_1\,w_1(x_1)\;\; b_2\,w_2(y_2)\;\; b_3\,r_3(x_0)\;\; c_1\;\; b_4\,r_4(y_0)\;\; c_2\;\; r_3(z_0)\;\; c_3\;\; w_4(z_4)\;\; c_4.
\]
In such a \emph{KTO-insensitive} window, any observed difference stems solely from the exclusion test itself; 
in practice, \textbf{ESSN} strictly dominates \textbf{SSN} with fewer false exclusions.

\paragraph{KTO-sensitive windows (anti-pivot structure).}
Real workloads also contain \emph{KTO-sensitive} windows where begin and commit orders diverge.
Consider the anti-pivot schedule \(m_2\); move \(c_2\) after \(c_4\) so that
\(b_1 \prec b_2 \prec b_3 \prec b_4\) while \(c_1 \prec c_3 \prec c_4 \prec c_2\):
\[
m_2' \;=\; b_1\,w_1(x_1)\;\; b_2\,w_2(y_2)\;\; b_3\,r_3(x_0)\;\; c_1\;\; b_4\,r_4(y_0)\;\; w_4(x_4)\;\; c_3\;\; c_4\;\; c_2.
\]
We obtain the anti-pivot structure \(t_1 \xleftarrow{b} t_3 \xrightarrow{f} t_4\), which is present under both
begin- and commit-ordered KTOs; only \(t_2\)'s relative position flips between the two.
With \( \mathrm{KTO}=\text{begin} \), \textbf{ESSN} admits the schedule while \textbf{SSN} tends to exclude
(aborting \(t_4\)). With \( \mathrm{KTO}=\text{commit} \), both ESSN and SSN succeed.

This illustrates that practical workloads \emph{mix} KTO-insensitive windows like \(M_1\) with KTO-sensitive ones
like \(m_2'\). Because \textbf{ESSN} is defined over the declarative triple
\((\mathrm{VF}, \mathrm{VO}, \mathrm{KTO})\), one can instantiate \(\mathrm{KTO}\) to match the workload at hand
(see \S\ref{sec:discussion}), while applying the same ESSN exclusion test unchanged.

\FloatBarrier

\section{Experiments and Evaluation}
\label{sec:experiments}
This section empirically evaluates ESSN against SSN with a focus on the abort rate of long
transactions. We first formalize what we call a long transaction and the evaluation metric. 
We then discuss the effect of KTO choice (begin vs. commit) on long aborts, which clarifies 
why subsequent experiments concentrate on commit-ordered KTO. Following this, we detail the single-threaded, 
reproducible history generator and checker used in our experiments, followed by the mixed workload 
design and key parameters. 
Finally, we present results under different read-from (RF) policies, which determine the VF, 
and explain when and why ESSN yields gains over SSN.
We conclude this section with a begin-ordered case study
(\Cref{subsec:mixed-begin}) that formalizes why a long never aborts
and how commit-stall aligns batch completion.

\subsection{What is a ``Long'' Transaction and How We Measure Abort Rate}
We call a transaction \emph{long} if its execution interval strictly contains at least one complete short transaction; 
e.g., in the schedule $m = b_1\, b_i\, c_i\, c_1$, $t_1$ is long and $\{t_i\}$ are short. 
In real workloads, a long transaction often contains multiple short transactions.

We measure the \emph{abort rate of a long transaction} as the ratio of the number of aborts to the number of trials 
of that specific long transaction.

\subsection{Effect of KTO on Long-Transaction Abort}
We begin by considering a case where several short transactions performing writes interleave within a long transaction:
\[
m_3 = b_1\ r_1(x_0)\ b_i\ w_i(x_i)\ c_i\ c_1
\]
In this case, the long $t_1$ aborts under OCC~\cite{kung1981-occ}, because a back-edge $t_1 \xrightarrow{rw} t_i$ 
is created.

Next, consider a case where both the long and the shorts are read-only, except that at the end the long performs 
a write that hits one of the keys read by a short:
\[
m_4 = b_1\ b_i\ r_i(x_0)\ c_i\ w_1(x_1)\ c_1
\]
Under MVTO, the long $t_1$ aborts in this schedule. With begin-ordered KTO, a back-edge $t_i \xrightarrow{rw} t_1$ arises, 
and MVTO sacrifices the write. Thus, even in a workload dominated by reads, selecting begin-ordered KTO can cause the long 
to abort. In essence, both OCC and MVTO are unfavorable to long transactions, because the relative begin/commit order 
between shorts and longs is constrained in a single direction. 
This is precisely where exclusion conditions that tolerate multiple back-edges, such as SSN or ESSN, become effective. 
For both $m_3$ and $m_4$, neither SSN nor ESSN aborts. 

However, in $m_3$, a long may read versions already committed by shorts, which can trigger SSN/ESSN exclusions; 
the situation is therefore more subtle. We revisit these cases in the experimental evaluation, where the difference 
between SSN and ESSN under $KTO{=}\text{commit}$ becomes evident.

\subsection{Experimental Setup}
\label{experimental_setup}
\paragraph{Reproducible single-threaded framework.}
We implement two components in Python~3 and run them in a notebook environment:
(i) a \emph{generator} that produces a multiversion schedule $m$ (begin/read/write/commit operations) 
according to a chosen RF policy, and
(ii) a \emph{checker} that, given a KTO and a VO aligned with the chosen KTO 
(begin- or commit-ordered), evaluates both SSN and ESSN exclusion rules.

\paragraph{History generator.}
First, we synthesize transaction histories that represent a controllable load pattern, 
and then materialize a multiversion schedule using common \emph{read-from (RF) policies}.
We consider three RF policies with respect to their alignment with begin- and commit-ordered KTOs (described below), 
and we implement the two that cover commit-KTO.

\begin{itemize}
  \item \texttt{as\_of\_read\_commit}: reads the latest \emph{committed} version at each read.
        \emph{Aligns with commit-KTO; may violate begin-KTO.}
        For example, in the schedule $b_1\, b_2\, \dots\, c_2\, c_1\, b_3$, a read by $t_1$ after $c_2$ may 
        observe $t_2$'s committed version, yielding a wr-edge $t_2\!\to\! t_1$ that is \emph{backward} under begin-KTO.

  \item \texttt{nearest\_begin\_kto}: reads from the nearest predecessor under begin-KTO (MVTO-like).
        \emph{Aligns with begin-KTO; may violate commit-KTO.}
        For example, in a similar schedule\\ $b_1\, b_2\, \dots\, c_2\, c_1\, b_3$, a read by $t_2$ after $b_1$ may 
        observe $t_1$'s write, yielding a wr-edge $t_1\!\to\! t_2$ that is \emph{backward} under commit-KTO.

  \item \texttt{snapshot\_at\_begin}: reads the latest committed version as of the reader's begin.
        \emph{Aligns with both begin- and commit-KTO.}
        Because the chosen writer both \emph{begins} and \emph{commits} before the reader's begin, 
        wr-edges are forward under either choice of KTO. 
\end{itemize}

\paragraph{History checker.}
Given a fixed KTO, the checker evaluates the SSN and ESSN exclusion tests at commit time under a per-item VO 
aligned with that KTO, deferring decisions as needed until dependent transactions have committed.
It does not model early aborts.
Our experimental implementation maintains per-key high-water marks for $\eta$ and $\xi$,
in place of the \emph{previous-edge-only} maintenance described in \S\ref{sec:implementation}.
Under single-threaded execution with commit-stalling, the two designs are equivalent:
the latest writer on a key observes the same $\eta/\xi$ values.
A discrepancy could arise only if a per-key high-water mark were to exceed the corresponding
prefix maximum; this cannot occur, because under the stall rule any reader whose position
in the KTO is later than the writer's (i.e., $\sigma(r)>\sigma(w)$) must wait for that
writer to commit before recording the read in the per-key high-water mark.

\subsection{Workload Design and Parameters}
\label{subsec:workload-design}
As illustrated in Fig.~\ref{fig:exp-scenario}, under commit-KTO with the \texttt{as\_of\_read\_commit} RF policy, 
the relative commit position of a short 
transaction determines whether its dependency with the long read-only $t_{long}$
is classified as a forward-edge ($t_{s2}$, $t_{s3}$) or a back-edge ($t_{s1}$, $t_{s4}$). 
This observation motivates our focus on commit-KTO in the experiments, 
as it directly influences whether SSN or ESSN identifies an exclusion violation.

\begin{figure}[t]
  \centering
  \includegraphics[width=0.75\linewidth]{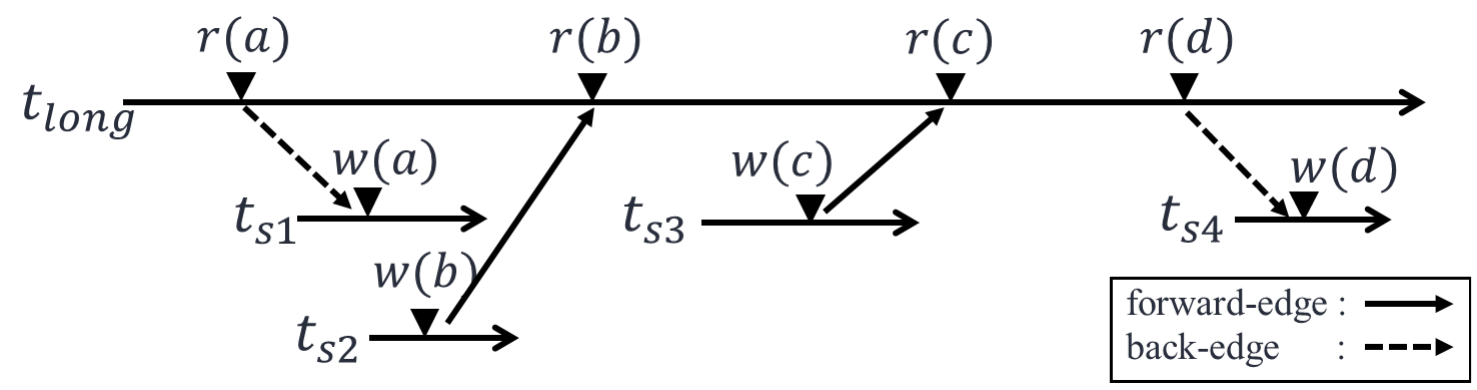}
  \caption{Illustration of dependency classification under commit-ordered KTO with 
  the RF policy \texttt{as\_of\_read\_commit}. 
  The relative commit position of short transactions determines whether the edge between the long 
  transaction $t_{long}$ and the short transaction becomes a forward-edge ($t_{s2}$, $t_{s3}$) or a back-edge ($t_{s1}$, $t_{s4}$).}
  \label{fig:exp-scenario}
\end{figure}

We use a mixed workload with two longs and many shorts:
\begin{itemize}
  \item $t_1$: read-only long (also reads a special item $z$),
  \item $t_2$: long that is mostly reads but performs a small final write to $z$ (e.g., summary/control table separate from $t_1$/$t_2$'s reads),
  \item $t_k$: write-only shorts that update random items except $z$. Shorts contend with each other; histories ensure some short starts within every $[b_k, c_k]$ interval.
\end{itemize}

Unless otherwise stated, we use:
\[
\texttt{N\_KEYS}=200,\quad
\texttt{READ\_SIZE}=40,\quad
\texttt{N\_SHORTS}=60,\quad
\texttt{REPEATS}=50,
\]
\[
\texttt{PIVOT\_PROBS}=\{0.0,0.2,0.5,0.8,1.0\},\quad
\texttt{SHORT\_HIT\_PROBS}=\{0.0,0.2,0.5,0.8,1.0\}.
\]

Here \texttt{PIVOT\_PROBS} controls the probability that $t_2$'s final write to $z$ overwrites a version read by $t_1$; 
\texttt{SHORT\_HIT\_PROBS} controls the probability that a short writes an item read by a long.
The generator creates shorts, embeds $t_1$/$t_2$, and then applies the selected RF policy to obtain 
a multiversion schedule. Each long starts after the short phase has begun, interleaves with shorts, 
and commits only after the last short that wrote any key in its read-set has committed.
We use the same generated multiversion schedule to evaluate both SSN and ESSN.

\subsection{Results}

\paragraph{R1: Average trends (Fig.~\ref{fig:bar_kto_commit}).}
Under a commit-ordered KTO, ESSN consistently lowers the long transaction's abort rate relative to SSN.
The effect is most pronounced with \texttt{snapshot\allowbreak\_at\allowbreak\_begin} reads, where the absolute abort rate of $t_2$ is low 
but the \emph{gap} between SSN and ESSN is large (in our runs, up to a 50\% relative reduction: $0.20 \to 0.10$).
With \texttt{as\allowbreak\_of\allowbreak\_read\allowbreak\_commit}, the absolute abort rate is higher and the ESSN--SSN gap narrows.
Fig.~\ref{fig:bar_kto_commit} summarizes these average trends.

\begin{figure}[t]
  \centering
  \includegraphics[width=0.6\linewidth]{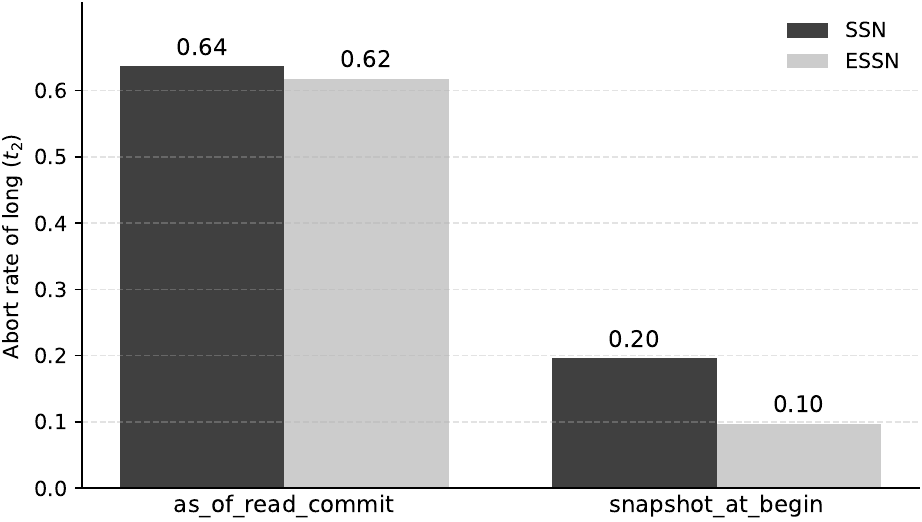}
\caption{Abort rate of the long transaction ($t_2$) under commit-ordered KTO across RF policies.
ESSN lowers the average abort rate versus SSN; the gap is largest with \texttt{snapshot\_at\_begin}.}
  \label{fig:bar_kto_commit}
\end{figure}

\paragraph{R2: Where the gains arise (Fig.~\ref{fig:gain_contours}).}
Fig.~\ref{fig:gain_contours} drills down by sweeping \texttt{PIVOT\_PROB} 
(probability that the long's final write hits a key it previously read) and \texttt{SHORT\_HIT\_PROB} 
(probability that shorts write keys read by the long).
The heat encodes the \emph{absolute} gain (SSN$-$ESSN), while the contour lines plot the abort rate of $t_2$ \emph{under SSN}.
Two patterns emerge:
(i) with \texttt{as\_of\_read\_commit}, contours run roughly horizontal, indicating sensitivity to \texttt{SHORT\_HIT\_PROB} 
and relatively weak dependence on \texttt{PIVOT\_PROB};
(ii) with \texttt{snapshot\_at\_begin}, vertical banding appears, indicating stronger dependence on \texttt{PIVOT\_PROB} 
via interactions with $t_1$. The largest absolute gap on the grid is about $0.25$, 
whereas the bar chart average in Fig.~\ref{fig:bar_kto_commit} is about $0.10$.

\begin{figure}[t]
  \centering
  \begin{subfigure}{0.49\linewidth}
    \includegraphics[width=\linewidth]{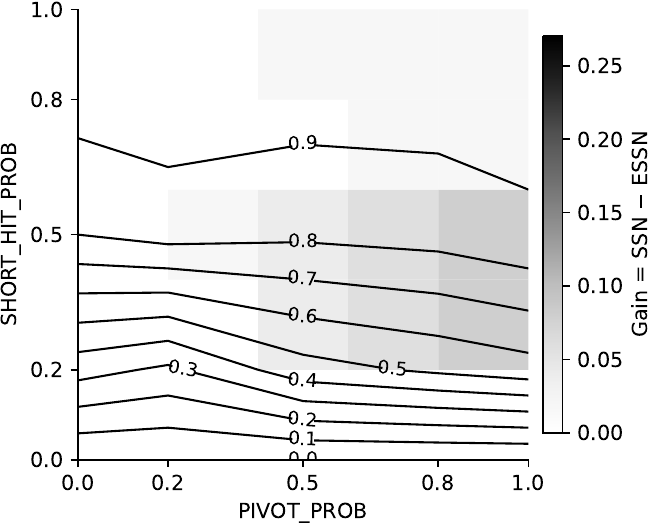}
    \caption{RF=\texttt{as\_of\_read\_commit}.}
    \label{fig:gain_contours:commit}
  \end{subfigure}
  \hfill
  \begin{subfigure}{0.49\linewidth}
    \includegraphics[width=\linewidth]{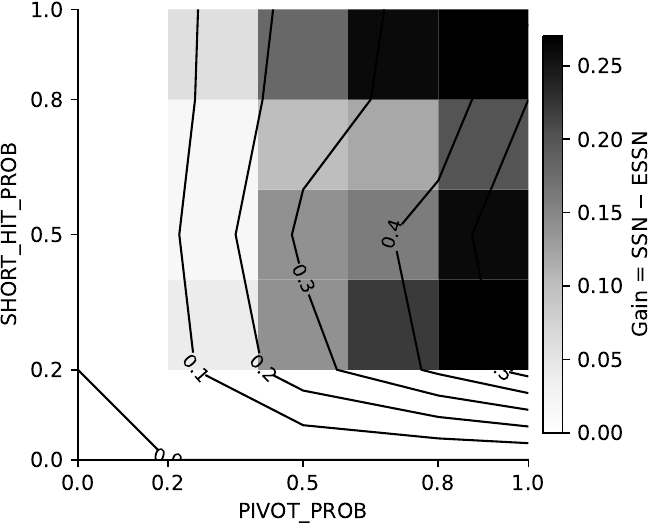}
    \caption{RF=\texttt{snapshot\_at\_begin}.}
    \label{fig:gain_contours:snapshot}
  \end{subfigure}
  \caption{Heatmap = ESSN gain (SSN$-$ESSN), \emph{contour lines = $t_2$ abort rate under SSN}. Grid over (\texttt{PIVOT\_PROB}, \texttt{SHORT\_HIT\_PROB}), commit-ordered KTO.}
  \label{fig:gain_contours}
\end{figure}

\paragraph{R3: Mechanistic explanation via the ESSN exclusion test.}
Recall that SSN aborts when $\pi(t)\!\le\!\eta(t)$, whereas ESSN aborts when $\pi(t)\!\le\!\xi(t)$.
For a long $t_2$, $\pi(t_2)$ is a \emph{minimum} over successors reachable via back-edges and stabilizes once established.
In contrast, $\xi(t_2)$ is a \emph{maximum} of $\pi(\cdot)$ over forward predecessors; it \emph{increases} 
as additional forward predecessors appear, which weakens ESSN's advantage 
(a larger right-hand side makes the inequality easier to satisfy).
Under \texttt{snapshot\_at\_begin}, short transactions that overwrite keys read by $t_2$ create back-edges not only to 
$t_2$ but also to $t_1$.
As these back-edges accumulate, they depress $\pi(t_1)$.
When $t_1$ later becomes a forward predecessor of $t_2$ ($t_1 \xrightarrow{f} t_2$), its small $\pi(t_1)$ enters
\[
\xi(t_2)=\max\{\pi(u)\mid u \xrightarrow{f} t_2\},
\]
keeping $\xi(t_2)$ low (often below $\eta(t_2)$), thereby reducing false positives and 
widening the ESSN--SSN gap observed with \texttt{snapshot\_at\_begin}.

\subsection{Mixed workload under begin-ordered KTO}
\label{subsec:mixed-begin}
\paragraph{Example (SSN/ESSN abort under commit-aligned order).}
Consider the schedule where $t_1$ is a long transaction and $t_2,t_3$ are short transactions:
\[
\mathtt{b_1\, b_2\, b_3\, r_3(y_0)\, r_1(x_0)\, w_2(x_2)\, c_2\, w_1(y_1)\, c_3\, c_1}\, .
\]
The MVSG edges are:
(i) $t_3 \to t_1$ on $y$ (read then later write), which is \emph{forward} under commit ordered KTO 
since $c_3 \prec c_1$;
(ii) $t_1 \to t_2$ on $x$ (read then later write), which is a \emph{back}-edge since $c_2 \prec c_1$.
Thus $t_1$ has an incoming forward and an outgoing back simultaneously; the SSN exclusion
($\pi(t_1)\!\le\!\eta(t_1)$) triggers, and SSN aborts $t_1$.
Note that with begin-ordered KTO, these orientations swap
(\emph{in-back} and \emph{out-forward} at $t_1$), so the pivot does not materialize and the same schedule
would be admitted.

\begin{lemma}[Long never aborts under begin-ordered ESSN/SSN]
\label{lem:begin-ordered}
Consider a mixed workload with one long transaction $t_\ell$ and many short transactions $\{t_i\}$.
Assume a KTO aligned with \emph{begin} times. Then under both ESSN and SSN, $t_\ell$ never aborts.
\end{lemma}

\begin{proof}[Proof sketch]
Under begin-ordered KTO, $t_\ell$ is earliest in the KTO.
For an \emph{rw} dependency between $t_\ell$ and a short $t$:
(i) if $t_\ell$ reads an item and $t$ later writes it, we get $t_\ell \!\to\! t$, which is a \emph{forward}-edge;
(ii) if $t$ reads an item and $t_\ell$ later writes it, we get $t \!\to\! t_\ell$, which is a \emph{back}-edge.
Hence every incoming \emph{rw} edge to $t_\ell$ is back and every outgoing \emph{rw} edge from $t_\ell$ is forward.
Since the SSN/ESSN pivot requires an outgoing back-edge from the pivot transaction (indeed, an incoming
forward \emph{and} an outgoing back), and $t_\ell$ has no outgoing back-edge under begin-ordered KTO,
$t_\ell$ never aborts under both ESSN and SSN.
\end{proof}

\paragraph{Priority control}
To ensure that the long transaction starts first (so that $\sigma(t_\ell) < \sigma(t_i)$ holds for all shorts),
it suffices to \emph{restart} any short transaction that happens to be active when the long begins.
This operationalizes the begin-ordered KTO assumption for mixed workloads.

\paragraph{Batch completion with commit-stall}
If we execute the long/short mix as a batch, the KPI is the completion time of the whole batch.
Under ESSN with begin-ordered KTO, \emph{commit-stall} waits are introduced to finalize
the stamps (e.g., $\pi$ and $\xi$) consistently with the KTO; in practice this means that short transactions
\emph{do wait}--not merely locally per key--until the necessary earlier commits become visible so that
their $\pi$/$\xi$ values are determined unambiguously.
Since (by Lemma~\ref{lem:begin-ordered}) $t_\ell$ never aborts, and all shorts ultimately wait to
synchronize their stamps with $t_\ell$'s earlier effects, the completion of all transactions effectively
aligns at $t_{\mathrm{commit}}(t_\ell)$ (shorts may retry or experience stalls, but the batch finish time
coalesces to $t_{\mathrm{commit}}(t_\ell)$; this is the begin-ordered counterpart of \S\ref{subsec:begin-stalls}).

\paragraph{Commit-stall suppresses the read-only anomaly.}
Replacing the short $t_3$ in the preceding example with a read-only $t_4$, consider the input trace
\[
\mathtt{b_1\, b_2\, b_4\, r_1(x{?})\, w_2(x_2)\, c_2\, r_4(x{?})\, r_4(y{?})\, c_4\, w_1(y_1)\, c_1}\, .
\]
With $\mathrm{VF}=\texttt{snapshot\_at\_begin}$, the multiversion schedule becomes
\[
\mathtt{b_1\, b_2\, b_4\, r_1(x_0)\, w_2(x_2)\, c_2\, r_4(x_2)\, r_4(y_0)\, c_4\, w_1(y_1)\, c_1}\, ,
\]
which induces a cycle $t_2 \xrightarrow{\mathrm{wr}} t_4 \xrightarrow{\mathrm{rw}} t_1 \xrightarrow{\mathrm{rw}} t_2$
---a classic read-only anomaly.
Under \emph{commit-stall} aligned with a begin-ordered KTO, stamps for the \emph{writer} $t_2$ are finalized
only after earlier effects become visible in order, whereas the read-only $t_4$ may \emph{bypass} the stall
(safe snapshot) and commit immediately. In our setting this resolves to:
\[
\mathtt{b_1\, b_2\, b_4\, r_1(x_0)\, w_2(x_2)\, r_4(x_0)\, r_4(y_0)\, c_4\, w_1(y_1)\, c_1\, c_2}\, ,
\]
which is MVSR with edges (ignoring the baseline $\mathrm{wr}$ from the initial version)
$\,t_1 \xrightarrow{\mathrm{rw}(f)} t_2,\;
t_4 \xrightarrow{\mathrm{rw}(b)} t_1,\;
t_4 \xrightarrow{\mathrm{rw}(b)} t_2$.
No cycle (or ESSN/SSN pivot) arises, so the schedule is admitted.

\paragraph{Relation to schedule-first enforcement.}
A recent system-level approach, MVSchedO~\cite{cheng2024-towards},
embodies a similar principle of begin-ordered alignment.
While ESSN formalizes the invariant that the earliest (long) transaction
never aborts because no outgoing back-edge can arise,
MVSchedO enforces an analogous ordering operationally:
transactions receive timestamps according to a precomputed schedule,
and later conflicting operations on hot keys are delayed
until earlier ones complete.

\section{Related Work}
\label{sec:related-work}

Concurrency control (CC) mechanisms are central to database systems, with an evolving focus on balancing serializability,
 performance, and implementation complexity. Among the many approaches, the handling of \emph{dependency cycles}---especially
  \emph{back-edges} in the serialization graph---has become a key axis for protocol comparison.

\subsection{Traditional Approaches}

\textbf{Optimistic Concurrency Control (OCC)}~\cite{kung1981-occ} and
 \textbf{Multiversion Timestamp Ordering (MVTO)}~\cite{weikum2001-tis} both assume a strict ordering
  based on timestamps. These protocols disallow back-edges and abort transactions that violate their predefined
   execution order, thereby ensuring serializability through re-execution.

\noindent\textbf{Multiversion Concurrency Control (MVCC)}, as implemented in systems like PostgreSQL, supports SI
 but can still admit non-serializable schedules unless extended with additional validation logic.

\subsection{Graph-based Protocols}

\textbf{Serializable Snapshot Isolation (SSI)}~\cite{fekete2005-ssi} introduced dependency tracking to detect cycles
 and anomalies. 

\noindent\textbf{Serial Safety Net (SSN)}~\cite{wang2017-ems} refined this idea by introducing \emph{exclusion windows}
 to efficiently prevent cycles in the serialization graph, allowing certain back-edges under carefully verified
  conditions. 

\noindent\textbf{Extended SSN (ESSN)} further reduces false positives while retaining lightweight global 
dependency tracking.

\subsection{Pragmatic Approaches}

\textbf{Cicada}~\cite{lim2017-cicada} adopts a pragmatic approach by combining MVCC with timestamp ordering, where timestamps are assigned at transaction start using loosely synchronized clocks.

\noindent Unlike TicToc, Cicada does not perform timestamp reordering; 
instead, it adheres to the MVTO-style principle of assigning fixed timestamps 
and aborting transactions whose access patterns violate version consistency. 
It tolerates temporal back-edges arising from anti-dependencies, but detects 
them from both the reader and writer sides at commit time, thereby preventing 
serialization cycles.

\noindent\textbf{Tsurugi/Shirakami}~\cite{kambayashi2023-tsurugi,kambayashi2023-shirakami} is a modern system
that adopts a begin-timestamp-based CC module (\emph{Shirakami}). From a graph-based perspective,
Shirakami corresponds to an approach that \emph{permits multiple back-edges} (rw anti-dependencies) while enforcing 
safety under begin-time ordering.
To avoid anomalies, Shirakami incorporates mechanisms such as \emph{order forwarding} and \emph{boundary alignment}.

\subsection{Dynamic Timestamp Protocols}

Recent protocols such as \textbf{TicToc}~\cite{yu2016-tictoc} allow dynamic adjustment of commit timestamps to resolve serialization anomalies. 
TicToc separates read and write timestamps for each data item and resolves dependencies dynamically at commit-time.
This mechanism allows the scheduler to resolve potential back-edges by \emph{timestamp reordering}, such that the final serialization order aligns with a legal commit order.
Although back-edges may appear during execution, they are neutralized by adjusting the commit timestamps, ensuring that the committed schedule is strictly serializable.

\subsection{Comparative Analysis}

A comprehensive empirical evaluation of major protocols is provided by CCBench~\cite{tanabe2020-ccbench}, 
which benchmarks OCC, MVTO, TicToc, SSN, Cicada, and others under unified workloads.
Table~\ref{tab:cc_comparison} complements those measurements by contrasting each protocol's stance on 
\emph{back-edges} and its \emph{timestamp/KTO basis}.
As seen in Table~\ref{tab:cc_comparison}, strict timestamp protocols preclude back-edges, graph-based checkers 
allow them with exclusion rules, and dynamic schemes adjust timestamps to neutralize them at commit.

\begin{table}[t]
  \centering
  \caption{Comparison of CC protocols by back-edge policy and timestamp/KTO basis.}
  \label{tab:cc_comparison}
  \begin{tabular}{lcc}
    \toprule
    \textbf{Protocol} & \textbf{Back-edge policy} & \textbf{Timestamp / KTO basis} \\
    \midrule
    OCC       & Not permitted          & Commit-time \\
    MVTO      & Not permitted          & Begin-time \\
    SSI       & Conditional            & Commit-time \\
    SSN       & Permitted              & Commit-time \\
    ESSN      & Permitted              & \textbf{Declared (KTO)} \\
    Shirakami & Permitted              & Begin-time \\
    Cicada    & Aborted if violated    & Begin-time (loosely synced) \\
    TicToc    & Neutralized by reorder & \textbf{Dynamic} \\
    \bottomrule
  \end{tabular}
\end{table}

\subsection{TicToc vs.\ MVCC (MVTO; SI with SSN/ESSN)}
\label{sec:tictoc-vs-mvcc}

TicToc represents a \emph{single-version} concurrency-control protocol:
it lacks multiversion visibility and effectively corresponds to a committed-read VF 
(\texttt{as\_of\_read\_commit}). Consequently, schedules that are multiversion-serializable (MVSR)
but not view-serializable (VSR) cannot be admitted by TicToc.
For example, even simple write-after-read or two-key read-skew patterns are rejected under single-version reasoning, yet they are safely admitted under MVSR (e.g., with $\mathrm{VF}=\texttt{nearest\_begin\_kto}$ or $\texttt{snapshot\_at\_begin}$).
Formal proofs for these cases are provided in App.~\ref{app:vsr-proofs}.
Furthermore, certain mixed long/short workloads (\S\ref{subsec:workload-design}) exhibit random read orders that can realize conflicting patterns, forcing aborts under $\mathrm{VF}=\texttt{as\_of\_read\_commit}$ but not under $\texttt{snapshot\_at\_begin}$ (Fig.~\ref{fig:bar_kto_commit}).
This demonstrates the conceptual gap between 
TicToc's single-version timestamp reordering and ESSN's
lightweight, KTO-aligned multiversion generalization.

\subsection{SSI vs.\ SSN vs.\ ESSN}
\S\ref{sec:ssn-limitations} establishes that
SSN strictly subsumes SSI under SI:
every SSN abort is also an SSI abort, but not vice versa.
Furthermore, \S\ref{sec:essn} proves that ESSN subsumes SSN.
Taken together, these results yield a strict hierarchy 
in acceptance-set terms under SI:
\[
A_{\text{ESSN}} \supset A_{\text{SSN}} \supset A_{\text{SSI}}.
\]
This hierarchy highlights that ESSN generalizes the SSN exclusion test while also capturing
SSI's dangerous-structure condition via a declared KTO.
In addition, SSI is limited to SI, and SSN is tied to a commit-ordered KTO,
whereas ESSN demonstrates broader acceptance power under any KTO-aligned criterion.

\section{Discussion}
\label{sec:discussion}

\subsection{Dual VF: Split Version Functions under a Single KTO}
\label{subsec:dual-vf}

\paragraph{Motivation.}
Prior work observes that MVCC read latency grows with the number of hops needed to
reach a snapshot-visible version on a per-key chain; each hop follows a \texttt{prev}
pointer and tends to miss caches, so even a few extra hops accumulate under high update
rates and parallelism~\cite{wu2017-mvcc-eval,lim2017-cicada,tanabe2020-ccbench}.

\paragraph{Design (within ESSN).}
Within our ESSN criterion (declared KTO; VF/VO aligned; prev-edge-only propagation),
we split VFs by workload role without introducing new assumptions:
(i) short transactions use a \emph{head-reading} VF (e.g., MVTO-style), so a read starts at the chain head
and returns the first visible version (often zero hops under a simple stall rule for later overwriters),
while (ii) long transactions use a snapshot VF certified by ESSN.
This keeps short reads near the head while preserving serializability for long ones.

\paragraph{Commit-time complexity.}
Because bounds are aggregated at read time (copy-up of successor/forward stamps) and writers
touch only their immediate predecessor, the ESSN exclusion check remains commit-only and
$O(|R|{+}|W|)$---often near $O(|W|)$ when most reads are dropped eagerly. No chain scan is required
during commit.

\paragraph{Takeaway.}
Structuring VFs so that most reads land at (or close to) the chain head cuts the dominant
per-read traversal cost, while a single global order and the ESSN exclusion test maintain
serializability even in mixed short/long workloads.

\paragraph{Relation to static robustness analysis.}
Our dual-VF design conceptually aligns with recent static robustness work that formalizes
mixed allocations over $\{\mathrm{RC}, \mathrm{SI}, \mathrm{SSI}\}$ under a commit-ordered
VO and corresponding VFs---RC reads relative to each operation and SI reads relative to
\texttt{first(T)}---and investigates when such mixed workloads remain serializable.
While orthogonal to ESSN's runtime safety net, these results provide a complementary,
offline viewpoint on mixed-isolation semantics similar to ours~\cite{Vandevoort2025-mvsplit},
and further suggest compatibility with begin-ordered reasoning and the dual-KTO model.

\subsection{Dual KTO: Leveraging Begin- and Commit-Ordered Reasoning}
\label{subsec:dual-kto}

\paragraph{Motivation.}
Beyond version functions, ESSN can also leverage flexibility in the choice of serialization order itself.
So far we have analyzed begin-ordered and commit-ordered KTOs separately.
We are aware of at least one practical CC mechanism that implicitly satisfies both at once:
\emph{snapshot isolation} (SI). SI fixes VF to ``begin=commit'' and also enforces
VO$=$begin$=$commit by disallowing concurrent $ww$ conflicts.
Thus, the version chain is the same under either begin- or commit-based reasoning,
and only the exclusion checks differ: one may evaluate $\pi$ and $\xi$ according to either KTO.

This raises the possibility of a \emph{dual KTO} approach.
At commit time, a transaction may commit as long as at least one of the exclusion conditions
(begin- or commit-based) does not trigger, effectively combining the advantages of both orders.
Even a naive deployment can exploit this flexibility by switching the KTO policy
per workload group according to observed characteristics.

These dual structures---for both VF and KTO---suggest a broader design space where
ESSN-style safety checks can coexist with adaptive or hybrid policies,
potentially bridging static and runtime serializability guarantees.

\section{Conclusion}

\textbf{ESSN} is a lightweight, invariant-driven framework for serializable MVCC that (i) factors correctness
into $(\mathrm{VF},\mathrm{VO},\mathrm{KTO})$ with explicit \emph{alignment}, and (ii) uses a single commit-time
exclusion test $\pi(t)\!\le\!\xi(t)$ that is sound for \emph{any} declared VO/KTO. The protocol is previous-edge-only
and achieves SSN-like cost---$\mathcal{O}(|R|{+}|W|)$, no chain scans, no early aborts---while \emph{strictly subsuming} SSN.

We further identified a localized \emph{anti-pivot} structure (\S\ref{subsec:anti-pivot-structure}) that often triggers SSN's
exclusion but is admitted by ESSN via forward-side $\pi$-propagation, irrespective of whether KTO is commit- or begin-ordered
(\S\ref{subsec:kto-general}). In mixed workloads under commit-ordered KTO, experiments (\S\ref{sec:experiments}) show
up to $\sim\!50\%$ relative reduction in long-transaction aborts (e.g., with \texttt{snapshot\_at\_begin}). 
When a long can be declared and scheduled with priority, begin-ordered KTO becomes a practical alternative; 
combined with dual-VF
(\S\ref{subsec:dual-vf}) and commit-stall, it maintains serializability while aligning batch completion.

Overall, ESSN is a drop-in, practical safety layer for MVCC engines across short-only, long-only, and mixed workloads.

\vspace{0.5em}
\subsection*{Summary of Contributions}
\begin{itemize}
  \item \textbf{Criterion and exclusion test.} Factor serializability into $(\mathrm{VF},\mathrm{VO},\mathrm{KTO})$,
  require VF/VO--KTO alignment, and provide an MVSG-based commit-time exclusion test; prove soundness and that
  ESSN strictly subsumes SSN.
  \item \textbf{Invariant-driven algorithm at SSN-like cost.} Present a DSG-based, previous-edge-only protocol with a single
  commit-time check that achieves $\mathcal{O}(|R|{+}|W|)$ work and no early aborts, matching SSN's per-version footprint.
  \item \textbf{Long/Long-RW as first-class targets.} Formalize interval-containment long transactions and show
  higher acceptance than SSN, with clear gains on Long-RW mixes.
\end{itemize}

\section{Acknowledgments}
This paper is based on results obtained from a project, JPNP16007, subsidized by the New Energy 
and Industrial Technology Development Organization (NEDO).

\appendix
\section{Proofs for \S\ref{sec:tictoc-vs-mvcc}}
\label{app:vsr-proofs}

\paragraph{Model assumptions (TicToc; single-version committed read).}
At the commit of a transaction $t$, TicToc must choose a commit timestamp $C_t$ that
simultaneously satisfies the following invariant\footnote{As usual, for reads the
implementation may raise $v.rts$ up to $C_t$ at $t$'s commit.}:

\begin{equation} 
  \begin{split} 
    \exists\,C_t:\quad
    \underbrace{\bigl(\forall v\in \{\text{versions read by }t\}\bigr)\; v.wts \le C_t \le v.rts}_{\text{(I\_1) read bounds}}\\
    \wedge\;\;
    \underbrace{\bigl(\forall v\in \{\text{versions overwritten by }t\}\bigr)\; v.rts < C_t}_{\text{(I\_2) write lower bound}} . 
  \end{split}
  \label{eq:tictoc-invariant}
\end{equation}

We write $x_{\mathrm{old}}$ for the version of key $x$ that existed when $t$ first read $x$.
In the schedules we consider there is at most one overwriter, so it suffices to distinguish
$x$ and $x_{\mathrm{old}}$.

\begin{lemma}[Single-key write-after-read is not allowed]
Consider the single-version schedule
\[
r_1(x)\;\; w_2(x)\; c_2\;\; w_1(x)\; c_1 .
\]
Then $t_1$ cannot commit under \eqref{eq:tictoc-invariant}.
\end{lemma}

\begin{proof}
From $r_1(x)$ and (I\_1) we obtain
\begin{equation}
C_1 \;\le\; rts(x_{\mathrm{old}}). \tag{A}
\end{equation}
From $w_2(x)$ and (I\_2) at $t_2$’s commit,
\begin{equation}
C_2 \;>\; rts(x_{\mathrm{old}}) \;\;\Rightarrow\;\; C_1 < C_2 \quad \text{by (A).} \tag{B}
\end{equation}
From $w_1(x)$ and (I\_2) at $t_1$’s commit we require $C_1 > rts(x)$; since after $c_2$ we have
$rts(x)\!\ge\!C_2$,
\begin{equation}
C_1 \;>\; C_2. \tag{C}
\end{equation}
(B) and (C) contradict. Hence $t_1$ cannot satisfy \eqref{eq:tictoc-invariant} and therefore cannot commit.
\end{proof}

\begin{lemma}[Two-key read skew is not allowed]
Consider the single-version schedule
\[
r_1(x)\;\; w_2(x)\; w_2(y)\; c_2\;\; r_1(y)\; c_1 .
\]
Then $t_1$ cannot commit under \eqref{eq:tictoc-invariant}.
\end{lemma}

\begin{proof}
From the early read $r_1(x)$ and (I\_1) we have
\begin{equation}
C_1 \;\le\; rts(x_{\mathrm{old}}). \tag{A}
\end{equation}
From $w_2(x)$ and (I\_2) at $t_2$’s commit,
\begin{equation}
C_2 \;>\; rts(x_{\mathrm{old}}) \;\;\Rightarrow\;\; C_1 < C_2 \quad \text{by (A).} \tag{B}
\end{equation}
Since $t_2$ overwrites $y$ before $r_1(y)$, the read of $y$ observes the version whose
write timestamp equals $C_2$; hence by (I\_1),
\begin{equation}
C_1 \;\ge\; wts(y) \;=\; C_2. \tag{C}
\end{equation}
(B) and (C) contradict. Therefore $t_1$ cannot satisfy \eqref{eq:tictoc-invariant} and cannot commit.
\end{proof}

\begin{lemma}[(a) is TicToc-feasible iff a bound interval is nonempty]
For
\[
\text{(a)}\quad r_1(x)\;\; w_2(x)\; c_2\;\; w_3(y)\; c_3\;\; r_1(y)\; c_1,
\]
there exists a feasible $C_1$ under \eqref{eq:tictoc-invariant} iff
\[
C_3 \;\le\; C_1 \;<\; C_2 \quad(\text{i.e., the interval }[C_3,\,C_2)\text{ is nonempty}).
\]
\end{lemma}

\begin{proof}
From the early read $r_1(x)$ and (I\_1),
\begin{equation}
C_1 \;\le\; rts(x_{\mathrm{old}}). \tag{A}
\end{equation}
From $w_2(x)$ and (I\_2),
\begin{equation}
C_2 \;>\; rts(x_{\mathrm{old}}) \;\Rightarrow\; C_1 < C_2 \quad \text{by (A).} \tag{B}
\end{equation}
From the late read $r_1(y)$ after $w_3(y),c_3$ and (I\_1),
\begin{equation}
C_1 \;\ge\; wts(y) \;=\; C_3. \tag{C}
\end{equation}
Thus a feasible $C_1$ exists iff \(C_3 \le C_1 < C_2\).
\end{proof}

\begin{lemma}[(b) is TicToc-feasible iff a bound interval is nonempty]
For
\[
\text{(b)}\quad r_1(y)\;\; w_2(x)\; c_2\;\; w_3(y)\; c_3\;\; r_1(x)\; c_1,
\]
there exists a feasible $C_1$ under \eqref{eq:tictoc-invariant} iff
\[
C_2 \;\le\; C_1 \;<\; C_3 \quad(\text{i.e., the interval }[C_2,\,C_3)\text{ is nonempty}).
\]
\end{lemma}

\begin{proof}
Analogous to the proof of (a), swapping $(x,\;t_2)$ with $(y,\;t_3)$.
\end{proof}

\begin{proposition}[Mutual incompatibility of (a) and (b)]
There is no initialization under which both (a) and (b) are admissible simultaneously.
\end{proposition}

\begin{proof}
By the lemma for (a), admissibility of (a) requires \(C_3 < C_2\).
By the lemma for (b), admissibility of (b) requires \(C_2 < C_3\).
These strict inequalities cannot both hold.
\end{proof}

\paragraph{Remark (VSR via serial order).}
A schedule is VSR iff there exists a total serial order $<_s$ in which every read observes
the last preceding write on the same key. For (a), $r_1(x)$ requires $t_1 <_s t_2$ while
$r_1(y)$ requires $t_3 <_s t_1$, yielding the unique order $t_3 <_s t_1 <_s t_2$.
For (b), $r_1(y)$ requires $t_1 <_s t_3$ while $r_1(x)$ requires $t_2 <_s t_1$, yielding
$t_2 <_s t_1 <_s t_3$. These required orders are incompatible, matching the
mutually exclusive interval conditions above.

\bibliographystyle{unsrt} 
\bibliography{essn_refs}

\end{document}